\documentclass[11pt]{article}
\usepackage{latexsym}
\usepackage{graphicx}
\usepackage{amsfonts}
\usepackage{amsmath, amsthm, amssymb}
\usepackage{fullpage}
\usepackage{setspace}
\usepackage{longtable}
\usepackage{natbib}
\bibpunct{(}{)}{;}{a}{}{,}
\usepackage{dcolumn}
\newcolumntype{.}{D{.}{.}{-1}}
\newcolumntype{d}[1]{D{.}{.}{#1}}
\usepackage{bm}
\usepackage{threeparttable}
\usepackage{booktabs}
\usepackage{enumerate}
\usepackage{bbm}
\usepackage{floatpag}
\usepackage{subfig}
\usepackage{url}
\usepackage[font=small]{caption}
\usepackage{enumitem}
\usepackage{xcolor}
\usepackage{multirow}
\usepackage{amsmath}
\usepackage{amsthm}
\usepackage{amssymb}
\usepackage{siunitx}
\usepackage{titlesec}
\usepackage[left=1in,right=1.05in,top=1in,bottom=1.05in]{geometry}
\usepackage[T1]{fontenc}
\usepackage{etoolbox}
\usepackage{changepage}
\usepackage{lscape}

\makeatletter
\patchcmd{\@maketitle}{\LARGE}{\large}{}{}
\makeatother

\titleformat{\subsection}[runin]
  {\normalfont\normalsize \bfseries}{\thesubsection}{0.5em}{}
  \titlespacing{\subsection}{1\parindent}{*0}{0.5em}
  
  \titleformat{\subsubsection}[runin]
  {\normalfont\normalsize \bfseries}{\thesubsubsection}{0.5em}{}
  \titlespacing{\subsubsection}{1\parindent}{*0}{0.5em}
  
  \titleformat{\section}
  {\normalfont\large \bfseries}{\thesection}{0.5em}{}
  \titlespacing*{\section}{0em}{0.5em}{0.5em}

\newtheorem{proposition}{Proposition}
\newtheorem{lemma}{Lemma}
\newtheorem{lemma*}{Lemma}

\newtheorem{definition}{Definition}

\newtheorem{theorem*}{Theorem}

\title{\vspace{-2em} \textbf{Re-evaluating the impact of hormone replacement therapy on heart disease using match-adaptive randomization inference}}

\author{Samuel D. Pimentel \&  Ruoqi Yu\thanks{Authors listed in alphabetical order (equal contribution). 
Samuel D. Pimentel is supported by the National Science Foundation under Grant No. 2142146. Part of this research was performed while Samuel D. Pimentel was visiting the Institute for Mathematical and Statistical Innovation (IMSI), which is also supported by the National Science Foundation (Grant No. DMS-1929348).}}
\date{}
\vspace{-2.5em}
\begin{document}

\maketitle


\begin{abstract}
Matching is an appealing way to design observational studies because it mimics the data structure produced by stratified randomized trials, pairing treated individuals with similar controls. After matching, inference is often conducted using methods tailored for stratified randomized trials in which treatments are permuted within matched pairs. However, in observational studies, matched pairs are not predetermined before treatment; instead, they are constructed based on observed treatment status.  This introduces a challenge as the permutation distributions used in standard inference methods do not account for the possibility that permuting treatments might lead to a different selection of matched pairs ($Z$-dependence).
To address this issue, we propose a novel and computationally efficient algorithm that characterizes and enables sampling from the correct conditional distribution of treatment after an optimal propensity score matching, accounting for $Z$-dependence.  We show how this new procedure, called match-adaptive randomization inference, corrects for an anticonservative result in a well-known observational study investigating the impact of hormone replacement theory (HRT) on coronary heart disease and corroborates experimental findings about heterogeneous effects of HRT across different ages of initiation in women. {\it Keywords:} matching, causal inference, propensity score, permutation test, Type I error, graphs.\end{abstract}

\doublespacing 
\section{Introduction}


A randomized trial is an ideal design for determining the causal effect of a treatment, since it ensures that units receiving the treatment and those receiving the control condition are comparable on average on all covariates, both observed and unobserved.  When only observational data is available, systematic differences between treated and control subjects must be addressed.  Matching designs proceed by pairing each unit receiving treatment to a control unit with similar observed covariates and comparing outcomes within the resulting matched pairs \citep{stuart2010matching}.  Intuitively, the aim is to create a new control subgroup that is similar 
 to the treated group on observed covariates, like the ideal control group in a hypothetical randomized trial. 
The analogy between matched designs in observational studies and stratified randomized trials extends to inference. It is common to permute treatment labels within matched pairs to conduct inference as has long been done for stratified randomized trials \citep{fisher1935design, rosenbaum2002observational}.  

However, permuting treatment labels in matched pairs from an observational study as if in a randomized trial relies on strong assumptions
, one of which is that pairs remain fixed even as treatment labels are permuted within them.  Although natural in stratified randomized trials where strata are chosen prior to treatment assignment, this assumption is suspect in matched observational studies, where matched pairs are explicitly a function of the observed treatment. 
Yet it has received little attention in the literature.  We demonstrate how ignoring the role of observed treatment in constructing matched pairs can lead to anticonservative permutation tests, and we propose a novel and computationally efficient inference algorithm that restricts attention to treatment permutations that would have produced the same set of matched pairs as the original study.  We call this new procedure match-adaptive randomization inference.  

We explore our new approach in a case study of a historically controversial question, the impact of hormone replacement therapy (HRT) on coronary heart disease (CHD) using observational data from the Women's Health Initiative study (WHI).  We show that match-adaptive randomization inference corroborates contemporary understanding, finding substantially less evidence for effects of HRT in the broad population than traditional methods of inference after matching, but confirming the presence of 
heterogeneous effects across different HRT initiation age groups. Overall, match-adaptive randomization inference provides a valuable safeguard in matched studies against subtle kinds of Type I error violations, calibrating inference procedures better to the ideal randomized trial.

\section{Background and literature review}

\subsection{The randomization inference paradigm in matched designs}

Matching offers several advantages compared to alternative approaches to adjusting for confounding, including a separation between a study's design stage 
and its analysis stage 
\citep{rubin2007design,rubin2008objective},
protection against extreme weights for individual subjects, 
and preservation of the original unit of analysis, which can facilitate 
complementary qualitative analysis \citep{rosenbaum2001matching,yu2021optimal}.  
At a high level, matching unites many beneficial attributes 
 of randomized trials 
	\citep{brown1980statistical}, and suggests the use of similar inference methods.   
Inference in a randomized trial can be conducted purely on the basis of the known distribution of treatment, holding covariates and potential outcomes fixed and requiring no distributional assumptions on outcomes besides the stable unit treatment value assumption (SUTVA), which prohibits  interference between subjects and 
hidden versions of treatment \citep{fisher1935design,rosenbaum2002observational}.  For stratified trials, this means permuting treatment indicators within strata.

While such ``randomization inference" methods are also applied to  matched observational studies in practice, using them in this context requires strong assumptions. In the language of  \cite{zhang2023randomization}, these are 
quasi-randomization tests.  We highlight three key assumptions:
(1) \emph{independent treatments}, requiring treatment status to be independent across individuals in the raw data prior to matching; 
(2) \emph{no unobserved confounding}, requiring the true probability of treatment 
to be a function of only measured covariates; 
and (3) \emph{exact matching}, under which matched units share identical propensity scores (conditional probabilities of treatment given covariates).
The first two assumptions, their possible violations, and the implications for a wide variety of causal inference approaches have received extensive discussion in the literature \citep{chang2022propensity,hansen2014clustered,rosenbaum1987sensitivity}. 
While more specific to the context of matching and less frequently discussed, the third assumption of exact matching {can also be} problematic.  
Even when propensity scores can be estimated consistently, 
exact matches are typically unattainable in practice except 
with few discrete covariates.  
 Failures of this assumption and the downstream adverse implications for inference are our primary focus.

\subsection{Addressing inexact matching}

Several researchers have recognized the implausibility of achieving exact matching on the true propensity score in finite samples and raised concerns regarding robustness of inference of violations to this assumption.  
\citet{hansen2009propensity} showed that in large sample settings where matched discrepancies 
shrink to zero, finite-sample bias shrinks slowly and may not vanish unless careful balance checking is applied; related work suggested the use of a shrinking propensity score caliper as a means to maintain accurate inference \citep{hansen2023matching}.  \citet{abadie2006large} found similar problems with inexact matching and proposed the use of regression models to re-establish valid inference \citep{abadie2011bias}.  \citet{saevje2021inconsistency} demonstrated that achieving convergence of matched discrepancies to zero becomes implausible when any units in the study population have propensity scores larger than 0.5. \cite{guo2022statistical} explored how this phenomenon leads to problems for randomization inference and suggested the use of regression models to rectify the problem.
However, caliper matching and regression modeling are imperfect solutions.  Calipers exclude treated subjects from the match, which hurts precision and alters the study population, making the study's estimand less interpretable.  Regression modeling requires assumptions about the study outcome, especially when matched pair discrepancies do not vanish asymptotically.

A different approach for correcting bias from inexact matching that does not require exclusion of treated units or outcome model assumptions is to modify the permutation procedure to account for propensity score differences within matched sets. Building on a brief initial proposal by \citet{baiocchi2011methodologies} and a growing recent literature on 
non-uniform permutation tests for unmatched studies  \citep{branson2019randomizationtests, berrett2020conditional, shaikh2021randomization}, \citet{pimentel2023covariate} 
introduced covariate-adaptive randomization inference, in which treatments are permuted within pairs with renormalized probabilities given by estimated propensity odds, 
which is valid even if  matched discrepancies do not vanish. Nevertheless, 
{these attractive properties of covariate-adaptive randomization inference rely on a key assumption that matched pairs are fixed across all permutations of treatment.} While reasonable for a randomization test in a paired randomized experiment (where matched pairs are determined prior to treatment assignment), this assumption is suspect in observational studies, where matched pairs are chosen based on the original observed treatment vector.  Consequently, some permutations of the treatment vector would have resulted in different matched pair configurations had they been observed originally. These treatment permutations do not appropriately belong to the conditional distribution of treatment, given that the original set of pairs was selected.  This phenomenon, labeled ``$Z$-dependence" by \citet{pimentel2023covariate}, was also discussed by \citet{pashley2021conditional}.  

In principle it is possible to remove incompatible permutations by repeating the propensity score match procedure for each permutation of treatment to see if the same pairs are recovered. However, in practice, this approach is computationally infeasible except in very small datasets, and more 
 efficient strategies are essential. 
In this paper, we introduce a novel and computationally efficient strategy called match-adaptive randomization inference to correct the permutation distribution and conduct valid hypothesis tests after matching.


\subsection{Hormone replacement therapy and women's health}\label{subsec:whi}

We use match-adaptive inference to study the effect of hormone-replacement therapy (HRT) on coronary heart disease (CHD) using data from the widely-recognized Women's Health Initiative study (WHI). Scientific debate over HRT's safety is decades old.  HRT was already widely used when the FDA approved estrogen for
 osteoporosis prevention in 1988 \citep{lobo2017hormone}, and during the subsequent period of even greater adoption, observational studies on HRT highlighted numerous benefits, including reduced risks of CHD and mortality \citep{grady1992hormone,grodstein1997postmenopausal,stampfer1991estrogen}. However, subsequent randomized trials, including one conducted as part of WHI, 
found no such benefit and raised concerns about increased risks of CHD and breast cancer \citep{anderson2004effects,herrington2000effects,hulley1998randomized,writing2002risks}. This divergence triggered a significant decline in HRT usage. Recent reanalyses of WHI data, with age stratification, along with comprehensive analyses of newer randomized trials and observational data, consistently demonstrate that when HRT is initiated shortly after menopause, there are notable reductions in CHD and mortality \citep{hodis2008postmenopausal,manson2013menopausal,salpeter2006brief,rossouw2007postmenopausal}. 
Unfortunately, these nuances in the data often went underreported in the media, leaving lingering apprehensions regarding HRT \citep{cagnacci2019controversial}.
Motivated by the debate over apparently contrasting results between the WHI's clinical trial results and previous observational studies  and by recent discoveries on the importance of timing and age at which HRT is initiated \citep{cagnacci2019controversial,lobo2017hormone}, we approach our analysis in \S\ref{sec:real_data} below from the perspective of matching and ask to what degree match-adaptive randomization inference is able to recalibrate uncertainty about the presence of a beneficial effect, in contrast to methods of inference that allow all within-pair permutations of treatment.

\section{Formal framework and role of $Z$-dependence}
\subsection{Setup for an inexactly-matched observational study}\label{subsec:setup}
Consider a sample of $N$ units independently drawn from a population. Each unit has an observed binary treatment indicator $Z$, an observed covariate vector $X$, and potential outcomes $Y(0)$ and $Y(1)$ which represent the values that would have been observed under each treatment option. Under SUTVA, we observe outcomes  $Y = ZY(1) + (1-Z)Y(0)$. 
 Let $ \mathcal{F} = (\mathbf{Y(0)}, \mathbf{Y(1)}, \mathbf{X})$ where boldface indicates the vector (or matrix) of all values in the sample.  
 {The investigator solves an optimization problem under which every subject in the treated group is paired to a distinct subject in the control group}, choosing pairs so as to minimize the average value of some predefined covariate distance between subjects \citep{rosenbaum1989optimal}; for a more formal characterization, see 
 \S\ref{app:matching} in the Supplementary materials. This produces $K \leq N/2$ disjoint pairs, indexed by $k$, each pair with two individuals $j = 1, 2$ so $Z_{k1} + Z_{k2} = 1$; {the remaining $N-2K$ individuals remain unmatched. 
We let $\mathcal{M}_{opt}$ represent the collection of $K$ index pairs for the matched units, and we let $\mathcal{U}_{opt}$ represent the collection of the indices   and treatment values for the  $N-2K$ unmatched units.} {Both $\mathcal{M}_{opt}$ and $\mathcal{U}_{opt}$ are deterministic functions} of $\mathbf{X}$ and $\mathbf{Z}$. 
{Our goal is} to conduct statistical inference for the matched sample using the distribution 
$\mathbb{P}(\mathbf{Z} \mid {\mathcal{M}_{opt}, \mathcal{U}_{opt}},  \mathcal{F})$. 

If unobserved confounding is absent so  $\mathbb{P}(Z=1 \mid X, Y(1),Y(0))$ is equal to the propensity score $ \mathbb{P}(Z = 1 \mid X) = \lambda(X) \in (0,1)$, 
and matched units share identical propensity scores (with mild conditions on choice among exact matches), $\mathbb{P}(\mathbf{Z} \mid {\mathcal{M}_{opt}, \mathcal{U}_{opt}},  \mathcal{F})$ is uniform over treatment permutation $\mathbf{Z}'$ {for which (i) $Z'_{k1} + Z'_{k2} = 1$ for all $k=1, \dots, K$ and (ii) unmatched units retain their original treatments, }  
i.e., for any such $\mathbf{Z}'$, $\mathbb{P}(\mathbf{Z}'|{\mathcal{M}_{opt}, \mathcal{U}_{opt}}, \mathcal{F})=1/2^K$. This  suggests permuting treatment labels with equal probability within matched pairs, while holding outcomes fixed, to conduct inference. We call this distribution $F_0(\mathbf{Z})$.  

In practice, it is  impossible to form matched pairs with identical propensity scores. 
Under inexact matching, the conditional distribution of treatment
$F_\lambda(\mathbf{Z})$ depends on 
propensity score differences between the two units in each matched pair.  Accordingly, \citet{pimentel2023covariate} suggested basing inference on the following distribution {$F_\lambda(\mathbf{Z})$} in place of $F_0(\mathbf{Z})$. {For any treatment permutation $\mathbf{Z}'$ such that $Z_{k1}' + Z_{k2}' = 1$ for all $k=1,\dots,K$ and all unmatched individuals' treatments are unchanged from the original $\mathbf{Z}$, we have:}  

\vspace{-3em}
\begin{align}
	\label{eqn:ppi}
	F_\lambda(\mathbf{Z})=
	 \prod_{k=1}^K\frac{\eta_{k1}^{Z_{k1}'}\eta_{k2}^{Z_{k2}'}}{\eta_{k1} +\eta_{k2}} 
	 = \prod_{k=1}^Kp_{k1}^{Z_{k1}}p_{k2}^{Z_{k2}},
\end{align}
\vspace{-3em}

\noindent where $\eta_{kj} = \lambda_{kj}/(1-\lambda_{kj})$ and $p_{kj} = \eta_{kj}/(\eta_{k1} + \eta_{k2})$. 
While the true propensity score $\lambda(\cdot)$ is typically unknown and must be estimated from observed data, these estimates can be used to approximate 
$F_\lambda$; 
\citet{pimentel2023covariate} proposed using this distribution for inference in practice, labeling this procedure covariate-adaptive randomization inference. 

Unfortunately, \citet{pimentel2023covariate}'s results ignore an important second source of Type I error violations. Permuting treatment labels using the 
distribution $F_\lambda(\mathbf{Z})$ 
essentially treats 
the matched pairs specified by $\mathcal{M}_{opt}$ as fixed across all {treatment permutation vectors $\mathbf{Z}'$ for which $Z'_{k1} +Z'_{k2} =1$ for all $k = 1, \ldots, K$ and unmatched units maintain their original treatment status}. However, in practice the selection of matched pairs depends on both $\mathbf{X}$ and $\mathbf{Z}$ and may vary across observed values of $\mathbf{Z}$.
Incorrectly including permutations that would have led to a different {choice} of $\mathcal{M}_{opt}$ can introduce bias into inference relative to the true conditional distribution for treatment after matching $\mathbb{P}(\mathbf{Z} \mid {\mathcal{M}_{opt}, \mathcal{U}_{opt}},\mathcal{F})$, particularly when propensity score differences correlate with outcome differences. 
 \cite{pimentel2023covariate} denote the discrepancy between $\mathbb{P}(\mathbf{Z} \mid {\mathcal{M}_{opt}, \mathcal{U}_{opt}},\mathcal{F})$ and $F_\lambda(\mathbf{Z})$ 
 as $Z$-dependence, in reference to $\mathcal{M}_{opt}$'s {and $\mathcal{U}_{opt}$'s} dependence on $\mathbf{Z}$, and document via simulation 
 how inference procedures based on 
  $F_\lambda(\mathbf{Z})$
  fail to control Type I error violations.
 

\subsection{$Z$-dependence and Type I error: a small example}
To illustrate how $Z$-dependence can lead to anticonservative inference in practice, we first consider a small example depicted in Figure~\ref{fig:small_example}. Four treated are units optimally matched into pairs (from a pool of six control units) using propensity score distances, yielding $\mathcal{M}_{opt}$ consisting of $\{(A,E), (B,G), (C,H), (D,I)\}$ and $\mathcal{U}_{opt}$ consisting of $\{F, J\}$ with their treatment indicators $\{0,0\}$. 
The sharp null hypothesis of no treatment effect for any individual holds, and study outcomes (identical to potential outcomes under control) are perfectly correlated with propensity scores. 
The difference-in-means statistic is 0.7. 

\begin{figure}
\vspace{-1em}
\small
		\centering
\begin{tabular}{ccc}
\small{(i) Original treatment assignment}& \small{(ii) Permute pairs A-E, B-G} & \small{(iii) Permute pairs B-G, C-H, D-I}\\[-2em]
\includegraphics[scale=0.45]{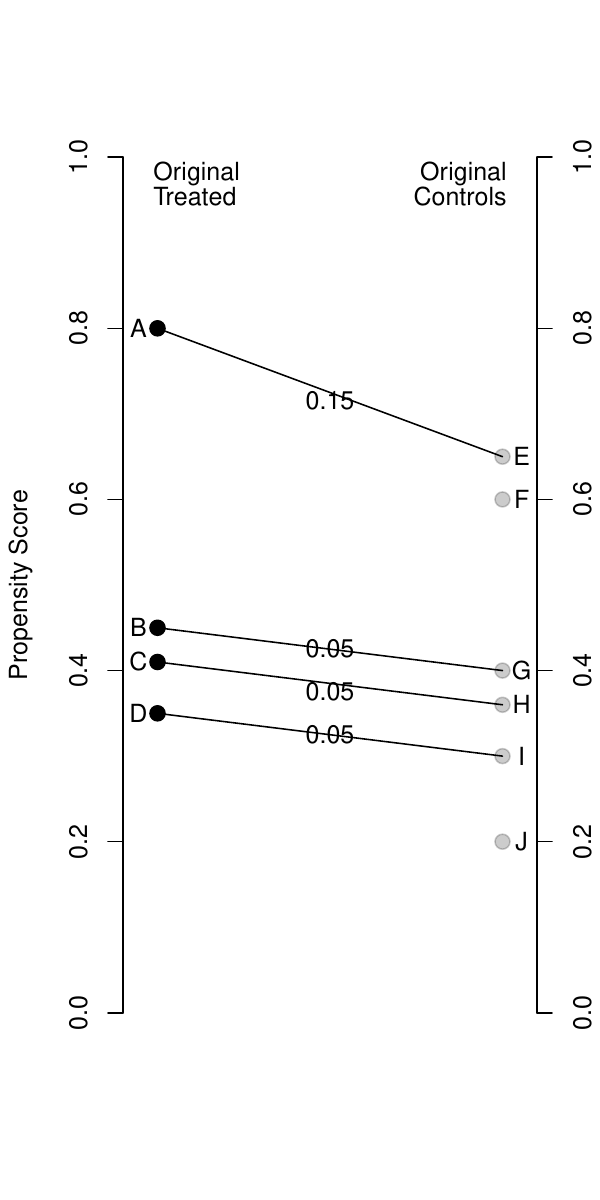}  & \includegraphics[scale=0.45]{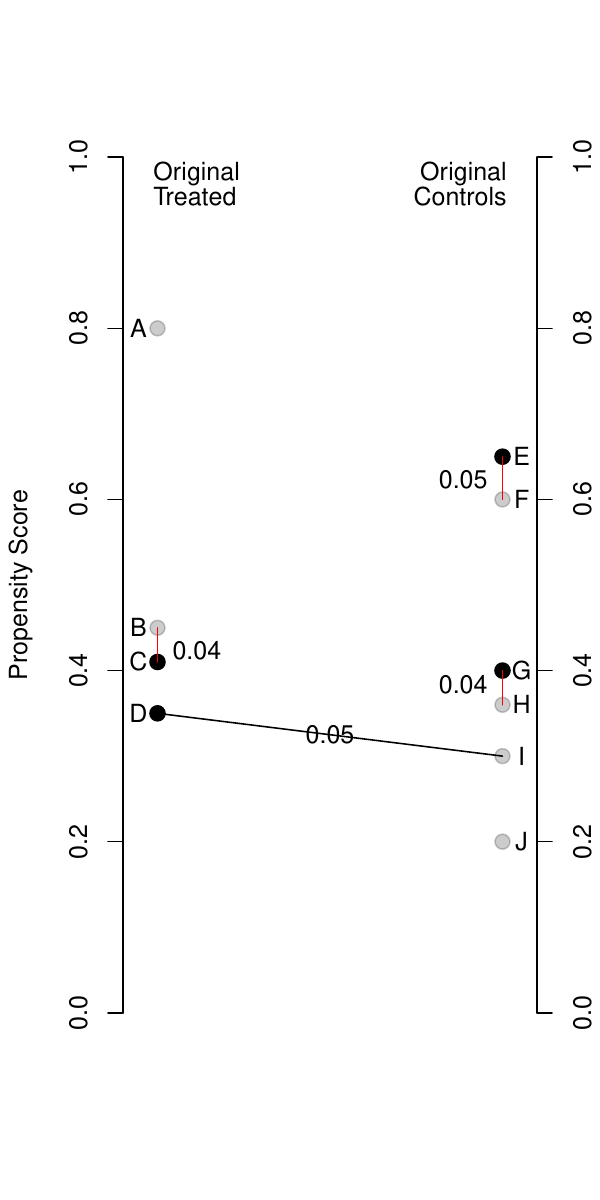} & \includegraphics[scale=0.45]{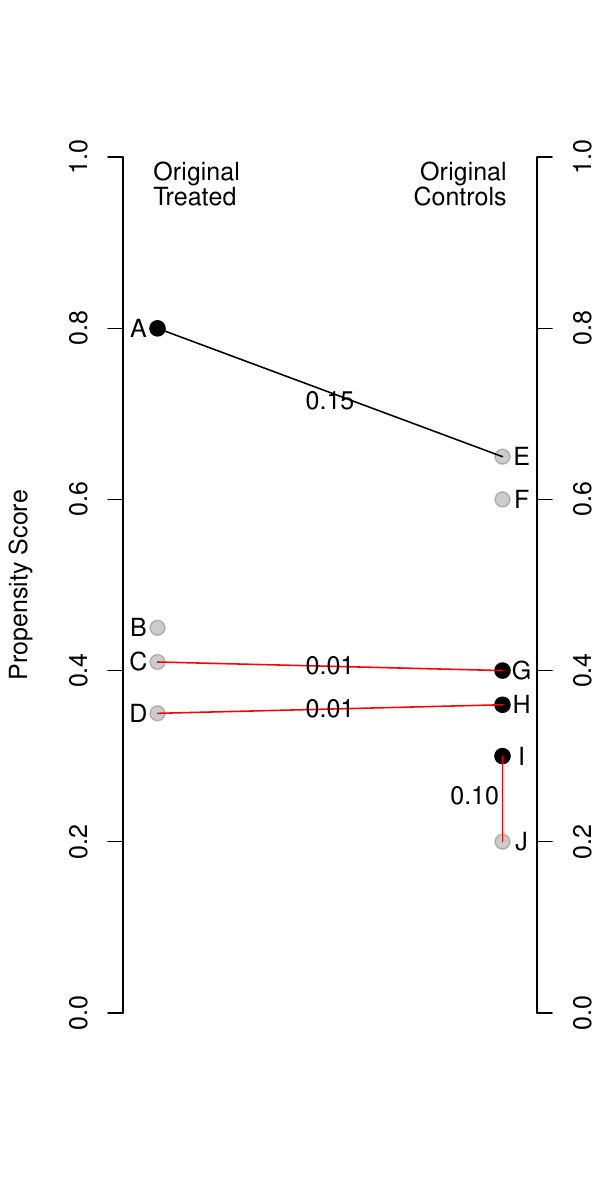}
\end{tabular}

\vspace{-3em}
\begin{tabular}{c|cccccccccc}
\hline
Unit & A & B& C& D& E &F &G &H & I & J\\\hline
Propensity Score & 0.80 & 0.45 & 0.41 & 0.35 & 0.65 & 0.60 & 0.40 & 0.36 & 0.30 & 0.20 \\ 
$Y(0)$ & 8.0 & 4.5 & 4.1 & 3.5 & 6.5 & 6.0 & 4.0 & 3.6 & 3.0 & 2.0\\\hline
\end{tabular}

\caption{\small
An optimal propensity score match and two permutations of treatment under which different matched pairs would have been formed if each had been observed initially. Black dots indicate treated units and gray dots indicate controls; black lines connect subjects in the same matched pair in the original match, and red lines show better matched pairs that could have been formed if this permutation had been observed originally. 
}
\label{fig:small_example}
\end{figure}

 With the uniform permutation test, 
a one-sided p-value of $1/2^{4} = 0.0625$ is obtained for the alternative hypothesis that the treatment effect is larger than zero. 
Next we perform a covariate-adaptive test using the same matched pairs. Since the treated units in each matched pair have higher propensity scores than their matched controls, the covariate-adaptive test recognizes that the original treatment allocation is relatively more likely than its permuted versions, and 
 the p-value increases to 0.11.  
 Finally, we account for $Z$-dependence by a brute force method, rerunning the optimal propensity score matching algorithm for each of the $2^4$ possible permutations. Since only three out of the 16 treatment permutations lead to the same set of matched pairs, we restrict attention to these permutations, renormalizing their covariate-adaptive probabilities to obtain a new distribution for treatment.  As a result, the p-value now increases further to approximately 0.41. Table \ref{tab:small_example} provides the complete distribution of the difference-in-means statistic under each treatment permutation.  While, in this particular case, none of the three tests rejects at the 0.05 level, in a slightly larger study, the first two tests might easily yield nominally significant results.

Evidently, conducting inference using the proper conditional distribution of treatment given the actual match constructed can protect against Type I error violations. 
This control comes at a cost in power -- in this example, for instance, it will be impossible ever to reject the null {at the typical 0.95 level (or even at the 0.25 level)} when accounting for $Z$-dependence since there are only three allowed permutations.  However, a test with good power is not very useful if it cannot also control Type I error  reliably.  As such, {we prioritize Type I error rate control over power considerations and focus on constructing a test that will be valid in the presence of $Z$-dependence. } 
 Unfortunately, the brute-force approach to adjusting for $Z$-dependence used for the small example scales very poorly with the size of the match, since the computational complexity of solving even one optimal match grows polynomially with the sample size, while the number of possible permutations of treatment (and hence distinct matches that need to be computed) explodes exponentially.

\begin{table}
\small
		\centering
\begin{tabular}{r|cccc|r|rrr}
	\hline
&\multicolumn{4}{c|}{Permute this pair?} & Diff.-in-Means & \multicolumn{3}{c}{Null probability} \\
&A-E & B-G & C-H & D-I & Value & Unif. & Cov.-Ad. & Match-Ad.\\
  \hline
1 &  &  &  &  & 0.70 & 0.06 & 0.12 & 0.41 \\ 
  2 & X &  &  &  & -0.05 & 0.06 & 0.05 &  \\ 
  3 &  & X &  &  & 0.50 & 0.06 & 0.09 &  \\ 
  4 & X & X &  &  & -0.25 & 0.06 & 0.04 &  \\ 
  5 &  &  & X &  & 0.50 & 0.06 & 0.09 &  \\ 
  6 & X &  & X &  & -0.25 & 0.06 & 0.04 &  \\ 
  7 &  & X & X &  & 0.30 & 0.06 & 0.08 & 0.27 \\ 
  8 & X & X & X &  & -0.45 & 0.06 & 0.04 &  \\ 
  9 &  &  &  & X & 0.45 & 0.06 & 0.09 & 0.32 \\ 
  10 & X &  &  & X & -0.30 & 0.06 & 0.04 &  \\ 
  11 &  & X &  & X & 0.25 & 0.06 & 0.08 &  \\ 
  12 & X & X &  & X & -0.50 & 0.06 & 0.03 &  \\ 
  13 &  &  & X & X & 0.25 & 0.06 & 0.07 &  \\ 
  14 & X &  & X & X & -0.50 & 0.06 & 0.03 &  \\ 
  15 &  & X & X & X & 0.05 & 0.06 & 0.06 &  \\ 
  16 & X & X & X & X & -0.70 & 0.06 & 0.03 &  \\ 
  \hline
\end{tabular}
\caption{Null distributions for the example in Figure \ref{fig:small_example} under uniform randomization inference, covariate-adaptive randomization inference, and match-adaptive randomization inference. The 16 rows delineate the $2^4=16$ possible permutations of treatments, with X's indicating which pairs were permuted.  Probabilities are given to two decimal places; in the final column, zero probabilities are omitted for greater clarity.}
\label{tab:small_example}
\end{table}

\section{Dealing with $Z$-dependence: match-adaptive randomization inference}
\label{sec:match_adaptive}
 
We now develop a computationally practical hypothesis testing framework using the distribution $\mathbb{P}(\mathbf{Z} \mid  {\mathcal{M}_{opt}, \mathcal{U}_{opt}}, \mathcal{F})$, 
labeled as match-adaptive randomization inference. This term signifies both its connection to covariate-adaptive randomization inference (which is based on the distribution  $F_\lambda(\mathbf{Z})$ 
 in the absence of exact matching) and its enhanced ability to handle the match itself as a random quantity. 
We characterize $\mathbb{P}(\mathbf{Z} \mid {\mathcal{M}_{opt}, \mathcal{U}_{opt}}, \mathcal{F})$ in three stages of increasing complexity: two matched pairs, many matched pairs without any unmatched controls, and a match with many pairs and unmatched controls.  In each stage we leverage a new structural property of the matched pair configuration to capture key information about compatibility of candidate treatment vectors with the original match.  

We assume for convenience that the matching distance is absolute difference on an estimated propensity score (although our results hold for matching on differences in any other univariate score). 
We assume there are no ties among estimated propensity scores (if some ties do occur, it is straightforward to introduce a small perturbation to obtain unique values). 
For extensions to handle exact matching  and caliper matching, see 
 Section~\ref{sec:gen}.

\subsection{Two pairs only: overlapping pairs and crossing matches}
\label{subsec:two_pairs}

First consider a toy problem with two treated units and two control units as in Figure \ref{fig:overlap_pairs}, so only two matched pairs are formed. 
In this case, only four permutations of treatment are possible.  To determine which of these permutations would have resulted in the same set of optimally matched pairs, we define a key concept, overlapping pairs. 

\begin{figure}
	\begin{tabular}{cc}(a)  \includegraphics[scale=0.65]{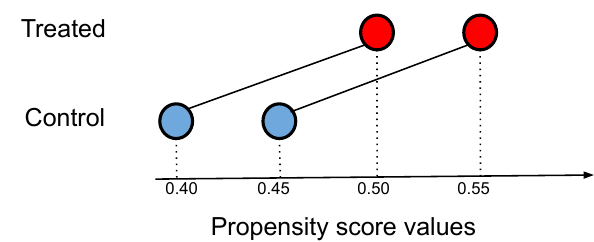} & (b) \includegraphics[scale=0.65]{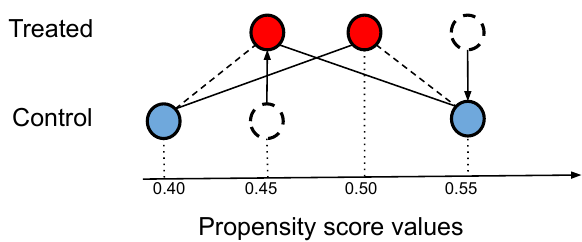} \end{tabular}
	\caption{Two overlapping pairs in (a), and a permutation of treatment within these pairs under which the fixed pairing becomes sub-optimal in (b).  The match in (b) can be improved in average propensity score difference by choosing instead the matches indicated by the dashed lines. }
	\label{fig:overlap_pairs}
\end{figure}

\begin{definition}
	\label{def:overlapping}
	Let $k_1$ and $k_2$ be indices of distinct matched pairs, and without loss of generality let 
	$\lambda(X_{k_11}) > \lambda(X_{k_12})$ and $\lambda(X_{k_21}) > \lambda(X_{k_22})$.
	Then, $(k_1,k_2)$ is an \emph{overlapping pair} if and only if
$	\lambda(X_{k_11})> \lambda(X_{k_21})> \lambda(X_{k_12}) \,\, \text{ or }\,\, \lambda(X_{k_21})>\lambda(X_{k_11})> \lambda(X_{k_22}).$
\end{definition}
\vspace{-0.5em}
 
A related concept is the crossing match introduced by \cite{saevje2021inconsistency}. For any two distinct matched pairs $k_1$ and $k_2$, suppose, without loss of generality, $Z_{k_11} = Z_{k_21} = 1$ in the original treatment vector.  Then, the pair $(k_1,k_2)$ is a crossing match if and only if
$\max\left\{\lambda(X_{k_11}), \lambda(X_{k_22})\right\} < \min\left\{\lambda(X_{k_12}), \lambda(X_{k_21})\right\}$. 
Crossing matches are suboptimal, since switching which treated and controls are matched 
leads to a strictly lower objective function in 
the optimal matching problem. Furthermore, an overlap relationship is present between two pairs if and only if they form a crossing match under some within-pair permutation of $\mathbf{Z}$.  This is why overlapping relationships matter for determining  $\mathbb{P}(\mathbf{Z} \mid {\mathcal{M}_{opt}, \mathcal{U}_{opt}}, \mathcal{F})$; all permutations in our two-pair study lead to the same optimal match if and only if the pairs do not overlap. If they do overlap, the only permutations in the support of $\mathbb{P}(\mathbf{Z} \mid{\mathcal{M}_{opt}, \mathcal{U}_{opt}},  \mathcal{F})$ are those in which either both pairs are permuted or neither pair is permuted. In Figure \ref{fig:overlap_pairs}(a), the two pairs overlap, and panel (b) illustrates how a crossing match arises when treatments in exactly one pair are permuted. 


\subsection{No unmatched units: connected components}
\label{subsec:no_unmatched}

When only two pairs are present and all units are matched, checking whether they overlap suffices to characterize  $\mathbb{P}(\mathbf{Z} \mid{\mathcal{M}_{opt}, \mathcal{U}_{opt}},  \mathcal{F})$ fully.  Now consider a more general setting with many matched units but no unmatched controls (i.e., equal numbers of treated and control subjects in the original observational study).  Overlapping pair relationships are still of central importance, and we represent them all
 a graph with nodes $\{1, \ldots, K\}$, each corresponding to a matched pair in $\mathcal{M}_{opt}$, and edges $(k_1, k_2)$ present only if $k_1\neq k_2$ and $k_1,k_2$ overlap. 
 The connected components of this graph 
  are the key structural element that determines our distribution of interest in this setting.  In brief, treatment permutations within a connected component (consisting of a group of one or more pairs linked by overlap relationships) must obey similar restrictions to a single set of overlapping pairs to be compatible with the original match. If at least one of the pairs in the connected component is permuted while at least one other is not permuted, a crossing match is necessarily created among the pairs and the original match is no longer optimal.  However, as long as unmatched units are absent, permutations that either keep or permute all treatment assignments within pairs lead to the same optimal match.   Figure \ref{fig:comps_example}(ii) shows connected components in the small example of Figure \ref{fig:small_example}.  The pairs  B-G and C-H are part of the same connected component since they overlap {(as illustrated in Figure \ref{fig:comps_example}(i))}.  
 
 \begin{figure}
 	\begin{minipage}{0.31\textwidth}
 		(i)\\ \includegraphics[scale=0.3]{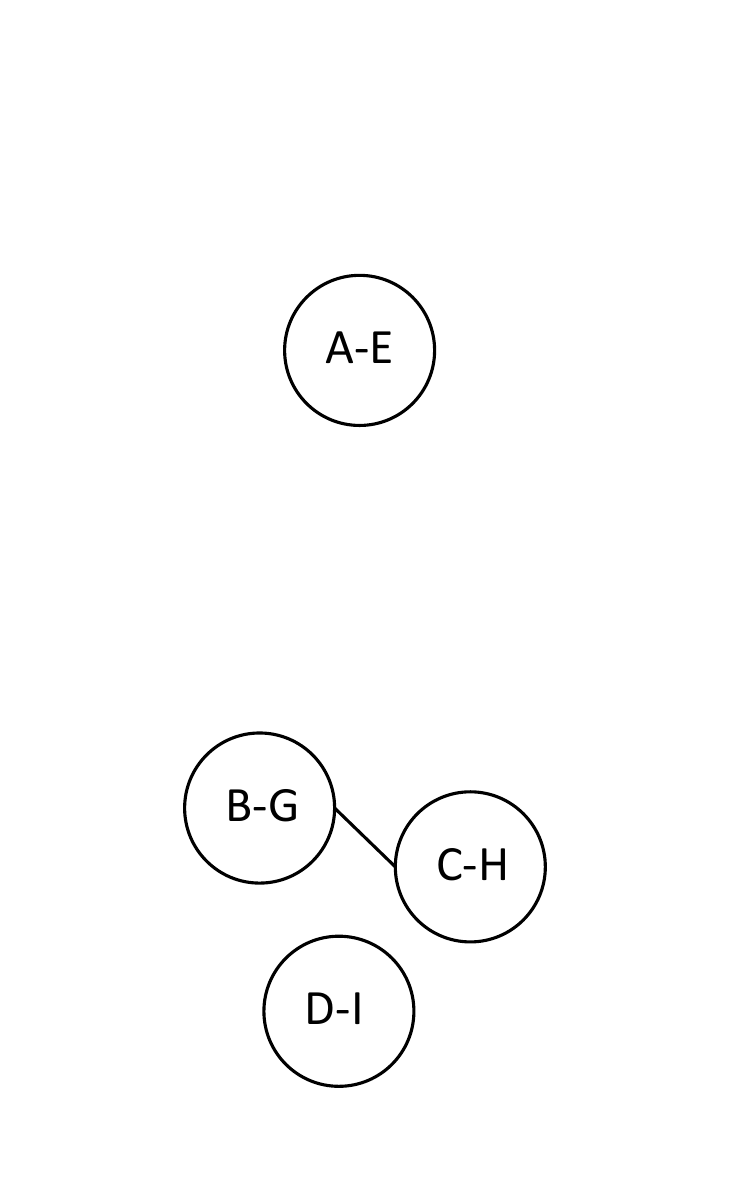}
 	\end{minipage}
 	\begin{minipage}{0.31\textwidth}
 		(ii)\\ \includegraphics[scale=0.3]{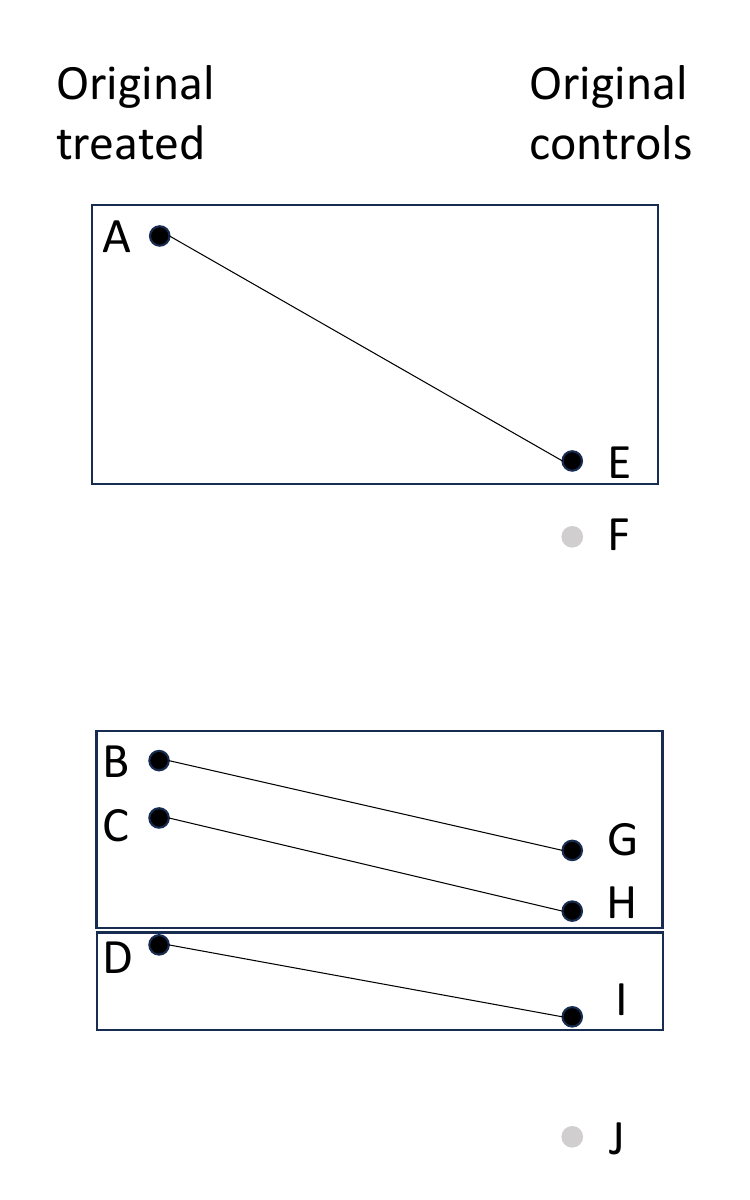}
 	\end{minipage}
 	\begin{minipage}{0.31\textwidth}
 		(iii)\\ \includegraphics[scale=0.3]{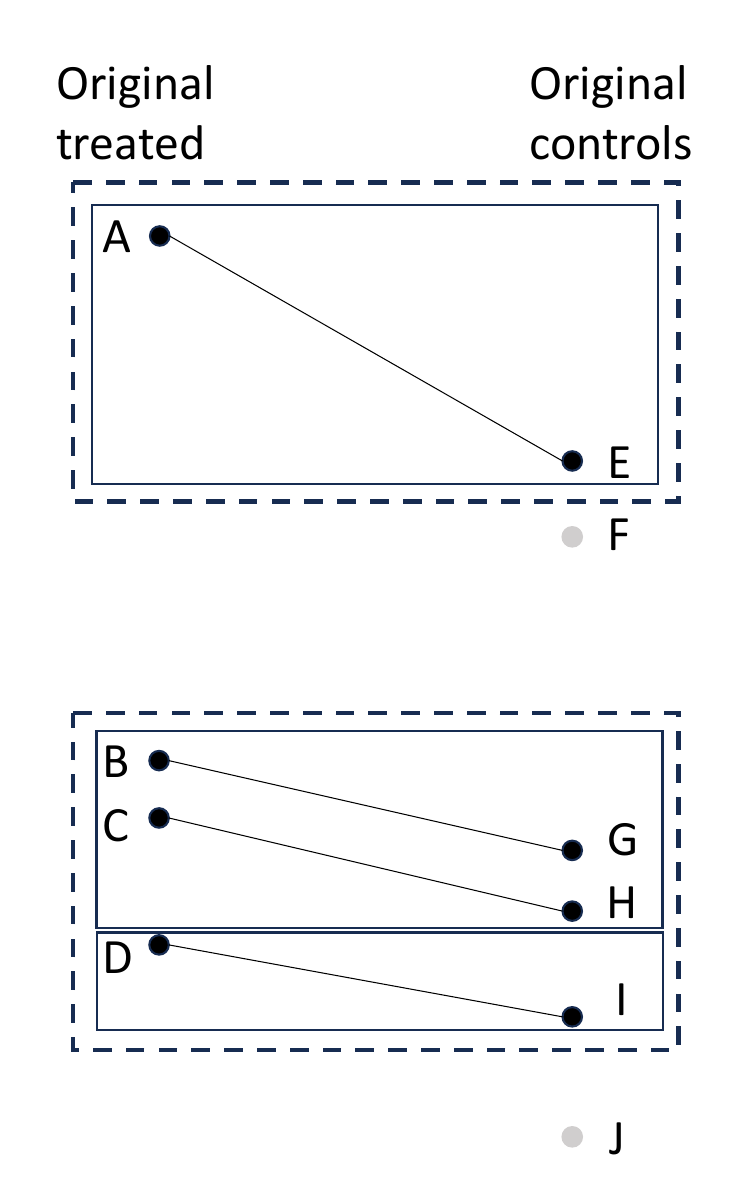}
 	\end{minipage}
 	\caption{Overlap graph for the pairs in the small example of Figure \ref{fig:small_example} (i), with further details on connected components (ii) and meta-components (iii).  In (i) each circle represents a pair (with unmatched individuals omitted), and the two pairs that overlap (B-G and C-H) are connected by an edge. In (i)-(ii), dots represent units rather than pairs.  Black dots represent matched individuals and gray dots represent unmatched controls with black lines connecting subjects grouped in the same matched pair in the match.
 		Solid rectangles each contain a distinct connected component of the overlap graph; dashed rectangles contain distinct meta-components.  In all figures, as in Figure \ref{fig:small_example}, objects are arranged vertically according to propensity score value. 
 	}
 	\label{fig:comps_example}
 \end{figure}

 Algorithm 1 provides an efficient procedure for 
  identifying the set of permutations that only permute entire connected components in lockstep, and Proposition \ref{prop:condsupport} shows that this modification suffices to identify the support of $\mathbb{P}(\mathbf{Z} \mid {\mathcal{M}_{opt}, \mathcal{U}_{opt}},\mathcal{F})$.  Subroutine \emph{ConstructOverlapGraph} is specified fully in \S\ref{app:subroutine} in the Supplementary materials. The proof of Proposition \ref{prop:condsupport} can be found in \S\ref{app:proof1} in the Supplementary materials.
%

 \textbf{Algorithm 1}
\begin{enumerate}[topsep=0.2pt, partopsep =0.2pt, itemsep=0.2pt, parsep = 0.2pt]
\item Conduct an optimal propensity score match with $K$ pairs. For each matched pair $k$, denote the treated propensity by $\lambda_k^T$ and control propensity score by $\lambda_k^C$.
\item Run the subroutine \emph{ConstructOverlapGraph} (with $(\lambda_1^T, \lambda_1^C), \ldots, (\lambda_K^T, \lambda_K^C)$ as inputs).

\item Identify the connected components of the graph and denote the indices for pairs in component $r$ by $\mathcal{S}_r, r=1,\dots, R$. 
\item Generate all permutation vectors $\mathbf{W} \in \{0,1\}^R$ where $W_r$ indicates whether the treatment assignments of all pairs in connected component $\mathcal{S}_r$ are reversed (relative to original treatment $\mathbf{Z}$).  The set of corresponding permuted treatment vectors $\mathbf{Z}'$ is returned as the support for the randomization test.
\end{enumerate}

\begin{proposition}
\label{prop:condsupport}
\textcolor{black}{Suppose that the initial match left no unmatched units.  Then for any compatible permutation $\mathbf{Z}'$ of $\mathbf{Z}$ generated by the algorithm, $\mathbb{P}({\mathcal{M}_{opt}, \mathcal{U}_{opt}}\mid \mathbf{Z}',  \mathcal{F}) =1$, and for any permutation $\mathbf{Z}'$ rejected, $\mathbb{P}( {\mathcal{M}_{opt}, \mathcal{U}_{opt}}\mid \mathbf{Z}',  \mathcal{F}) =0$. } 
\end{proposition}

Algorithm 1 runs very efficiently. \emph{ConstructOverlapGraph} runs linearly in the number of pairs. Finding the connected components of a graph with $V$ nodes and $E$ edges runs linearly in $E + V$ \citep{hopcroft1973algorithm}; since in the worst case the number of overlapping relationships scales as $K^2$ with the number of pairs $K$, this is $O(K^2)$.  Both tasks tend to be less computationally expensive than constructing the optimal match itself.

The efficient method described in Algorithm 1 provides a way of choosing the support for a randomization test (i.e., which permutations will be allowed) but does not describe probabilities for the valid permutations.  These can be calculated by re-normalizing the probabilities from the covariate-adaptive randomization test of \citet{pimentel2023covariate} to account for the altered support.  Letting $\mathcal{S}_r$ represent the indices for pairs in component $r$ as in Algorithm 1, 
we conduct the test using the following probabilities for the $W_r$s: 

\vspace{-3em}

\begin{align}
\mathbb{P}(W_r = 1)  &= \frac{\theta_r^{switch}}{\theta_r^{switch} + \theta_r^{keep}}, \quad  \quad
\theta_r^{keep} = \prod_{k \in \mathcal{S}_r}p_{k1}^{Z_{k1}}p_{k2}^{Z_{k2}}, \quad 
\quad
\theta_r^{switch} = \prod_{k \in \mathcal{S}_r}p_{k1}^{1-Z_{k1}}p_{k2}^{1-Z_{k2}}. 
\label{eqn:w_prob}
\end{align}

\vspace{-1em}
\noindent Here, $p_{k1}$ and $p_{k2}$ are defined 
as in equation (\ref{eqn:ppi}).

\subsection{General case with unmatched controls: meta-components}

In practice, we usually have a larger pool of potential controls than the treated units, so that optimal pair matching leaves some control subjects unmatched.  The presence of these unmatched controls leads to complications when determining the distribution $\mathbb{P}(\mathbf{Z} \mid {\mathcal{M}_{opt}, \mathcal{U}_{opt}}, \mathcal{F})$. Although not useful for matching under the original treatment $\mathbf{Z}$, these controls might be included in the optimal matched control group under some permutations of treatment.  For example, compare panels (a) and (b) of Figure \ref{fig:small_example}; although the matched pair A-E does not overlap with any other pair, permuting it permits E and F to match to each other instead with lower cost.   Thus, Algorithm 1, which examines reconfiguration only among subjects included in the original match, is not enough.  

In fact, the impact of unmatched controls can lead to subtler problems that extend beyond a single matched pair or connected component.  Consider the matched pairs C-H and D-I in Figure \ref{fig:small_example}.  These pairs do not overlap so Algorithm 1 places no restrictions on how they can be permuted.  If either one of these pairs is permuted while the other remains unchanged, the original match remains optimal.  However, as panel (c) of Figure \ref{fig:small_example} shows, it is not permissible to permute both C-H (along with pair B-G in the same connected component) and D-I, since this opens the possibility of a better match involving the previously unmatched control J.  This shows that when unmatched controls are present, it is no longer sufficient to consider connected components in isolation, as Algorithm 1 did.  The presence of an unmatched control can create entanglement among the allowed permutations of several neighboring connected components.

To address these challenges, we leverage a final structural property of a set of matched pairs, the meta-component.  Call two connected components $r, r'$ adjacent if there are no unmatched controls or other connected components between them in propensity score space.  Formally,  
adjacency holds if for any unit $k$ in $\mathcal{S}_r$ and any unit $k'$ in $\mathcal{S}_{r'}$, each unit $k''$ with propensity score $\lambda_{k''}$ between $\lambda_{k}$ and $\lambda_{k'}$ lies in either $\mathcal{S}_r$ or $\mathcal{S}_{r'}$. 
For any connected component, its meta-component is the largest group of consecutive adjacent connected components to which it belongs. By default, connected components not adjacent to any other components are also considered to be meta-components of size 1.  Figure \ref{fig:comps_example}(iii) shows the meta-components in the match from the small example of Figure \ref{fig:small_example}.

Meta-components are useful because they define how far across the matched design as 
we have to look to consider the possible impact of including a single new unmatched control.  As will be shown, if under a treatment permutation none of the unmatched controls that lie closest to the edge of any meta-component can be used to improve the match among units in that meta-component, the current match remains optimal.
Accordingly, Algorithm 2 leverages meta-components to identify the correct support for match-adaptive randomization inference in the presence of unmatched controls.
It treats each meta-component separately, recursively exploring all possible permutations of treatment for the internal components of each and screening out permutations under which unmatched controls can enter the match and improve the objective.  All subroutines are specified 
in \S\ref{app:subroutine} in the Supplementary materials. 

\textbf{Algorithm 2}
\begin{enumerate}
\item Build overlap graph and its connected components (steps 1-3 of Algorithm 1).
\item Run subroutine \emph{ConstructMetaComponents} (with inputs $\lambda_1, \ldots, \lambda_N$ and $\mathcal{S}_1, \ldots, \mathcal{S}_R$) and store output as $\mathcal{O}$. Denote the number of meta-components as $I$.
\item (Preprocess treatment assignments for each meta-component) For each $i =1, \ldots, I$:
\begin{enumerate}
%
%
%
\item Run subroutine  \emph{CheckComponentUp}$(\mathcal{O},1,0,0,s_i^\ell, \mathbf{w})$  for $\mathbf{w}=(0,...,0)$ and $\mathbf{w}=(1,0,\ldots,0)$.  Label the resulting set of compatible vectors $\mathcal{W}_{i,up}$.  
%
\item Run \emph{CheckComponentDown}$(\mathcal{O}, |\mathcal{T}_i|,0,0,s_i^u, \mathbf{w})$ for $\mathbf{w}=(0,...,0)$ and $\mathbf{w}=(1,0,\ldots,0)$.  Label the resulting set of compatible vectors, $\mathcal{W}_{i,down}$. 
\item Define $\mathcal{W}_{i} = \mathcal{W}_{i,up} \cap \mathcal{W}_{i,down}$.  
\end{enumerate}

\item Draw $\mathbf{W}$ by sampling $\mathbf{w_i}\in\mathcal{W}_{i}$ for meta-component $i$. 

\item Choose treatment permutation $\mathbf{Z}'$ by setting all subjects in each component to the corresponding treatment assignment in $\mathbf{W}$.
\end{enumerate} 

Validity of Algorithm 2 is stated in Proposition~\ref{prop:unmatched} (see \S\ref{app:proof2} in the supplementary materials for the proof).
\begin{proposition}\label{prop:unmatched}
	For any permutation $\mathbf{Z}'$ of $\mathbf{Z}$ rejected by Algorithm 2, $\mathbb{P}({\mathcal{M}_{opt}, \mathcal{U}_{opt}} \mid \mathcal{F}, \mathbf{Z}') = 0$, and for any permutation $\mathbf{Z}'$ not rejected, $\mathbb{P} ({\mathcal{M}_{opt}, \mathcal{U}_{opt}}\mid \mathcal{F}, \mathbf{Z}') =1$.
\end{proposition}

Once the conditional distribution of treatment within each meta-component has thus been characterized, we draw permutations for each meta-component independently of others, with re-normalized covariate-adaptive probabilities derived from the probabilities for individual connected components given in equation (\ref{eqn:w_prob}).  For any set of assignments $\mathbf{w}_i$ over the connected components in meta-component $i$ not rejected by Algorithm 2 (i.e., $\mathbf{w}_i \in \mathcal{W}_i$):
\begin{equation}
\mathbb{P}_{norm}(\mathbf{w}_i) = 
\frac{ \mathbb{P}(\mathbf{W}_i = \mathbf{w}_i)}{\sum_{\mathbf{w}_i' \in \mathcal{W}_i}\mathbb{P}(\mathbf{W}_i = \mathbf{w}'_i)}\mathbf{1}\left\{\mathbf{w}_i \in \mathcal{W}_i \right\}. 
\label{eqn:alg2probs}
\end{equation}
Since permutations  are independent across meta-components, these probabilities can be multiplied to derive sampling probabilities for full treatment vectors.

\subsection{Implications for hypothesis testing and confidence intervals}

So far we can accurately sample from the conditional distribution of treatment $\mathbb{P}(\mathbf{Z} \mid {\mathcal{M}_{opt}, \mathcal{U}_{opt}},  \mathcal{F})$ via appropriately constrained permutations, at least when propensity scores are known.  This means that match-adaptive randomization tests based on our constrained permutation approach with known propensity scores will control Type I error. To make this claim formal, we introduce some machinery for hypothesis testing.  Let $T:(\mathbf{Z},\mathbf{Y}) \longrightarrow \mathbb{R}$ be an arbitrary test statistic and let $\alpha \in (0,1)$ be the desired level of Type I error control.  Let $\mathbf{Z}'$ be an arbitrary draw of treatment assignments based on Algorithm 2 with probabilities given by equations (\ref{eqn:w_prob})-(\ref{eqn:alg2probs}), using the true propensity scores, and let $t(\alpha) = \min\{t: P[T(\mathbf{Z}',\mathbf{Y}) > t \mid  {\mathcal{M}_{opt}, \mathcal{U}_{opt}}, \mathcal{F}] \leq \alpha\}$.  
Using the quantile $t(\alpha)$ of our match-adaptive randomization distribution as a critical value for a one-sided test of a sharp null hypothesis of zero effect controls Type I error at level $\alpha$. The formal results are summarized in Proposition~\ref{prop:typeI}; see \S\ref{app:proof3} in the Supplementary materials for the proof. 
 
\begin{proposition}\label{prop:typeI}
Suppose the sharp null hypothesis $Y_{ki}(1)= Y_{ki}(0)$ for all $k,i$ holds.  Then 

$P[T(\mathbf{Z},\mathbf{Y}) > t(\alpha) \mid{\mathcal{M}_{opt}, \mathcal{U}_{opt}}, \mathcal{F}] \leq \alpha$ 
for all $\alpha \in (0,1)$.
\end{proposition}


We may also invert the match-adaptive randomization test to provide valid $1-\alpha$ confidence intervals.  Consider the null hypothesis of a constant additive effect of size $\tau$, 
$Y_{ki}(1) - Y_{ki}(0) = \tau.$
We can test this null hypothesis instead of the $\tau=0$ considered in Proposition \ref{prop:typeI} by replacing $\mathbf{Y}$  with $\mathbf{Y}-\tau \mathbf{Z}$ in the hypothesis testing framework described above.  A $1-\alpha$ confidence set is then obtained by choosing all values $\tau$ for which the test does not reject at level $\alpha$; as long as true propensity scores are used, the $1-\alpha$ coverage follows as a consequence of Proposition  \ref{prop:typeI}.  For more details on constructing confidence intervals 
see \citet{luo2021leveraging}. 


\section{Practical Considerations}\label{sec:gen}

\subsection{Speed-ups for large meta-components}
\label{subsec:speedup}
Algorithm 2 calculates the treatment distribution by examining all possible treatment assignments in each meta-component.  For large meta-components it can become inefficient to explore all possible treatments.  Large datasets need not contain large meta-components 
since many pairs may be subsumed either into large connected components or into a large number of small meta-components.  However, in practice we recommend checking for very large meta-components before running Algorithm 2. 

Some changes to the algorithm can help optimize performance for large meta-components. If one is approximating the overall distribution by some finite number of draws $N_{sim}$, a quicker version of the algorithm avoids exploring all possible treatment assignments, instead first sampling treatment permutations from the incorrect distribution, then checking which are valid and keeping only those (adding extra draws as needed until at least $N_{sim}$ are obtained).  For full details see Algorithm 2' in \S\ref{app:large} in the Supplementary materials.

Early stopping of the recursive procedures  \textbf{CheckComponentUp} and \textbf{CheckComponentDown} can also help.  These subroutines exhaustively explore treatment assignments within a meta-component, 
tracking possible changes to the objective function as matched pairs are perturbed one by one.  Sometimes at an early stage the objective function has increased by so much that it is mathematically impossible for treatment assignments further down the recursive path to improve on the original match, in which case the recursive search can be stopped. 
For more details see Algorithm 3 in \S\ref{app:caliper} in the Supplementary materials.


\subsection{Extensions to exact matching and caliper matching}
\label{subsec:exact_caliper}
While enhancing computational performance for the inference procedure is helpful in practice, often a bigger concern is ensuring efficient computation of the initial match.  For example, in the matched study comparing HRT users to controls in Section \ref{sec:real_data}, the match required 14 times as much computation time as the inference procedure.  
For reasons discussed by \citet{pimentel2015large} and \citet{yu2020matching}, matches can often be computed more efficiently when matching is exact on
  one or more categorical variables, 
   or when a caliper is imposed on one or more continuous variables.
  It is straightforward to adapt match-adaptive inference to exact matching constraints, since this amounts to performing several separate matches within distinct categories.  Algorithm 2 can be repeated separately within each category to generate a valid permutation distribution for that segment of the overall match;  since matches across categories are forbidden an overall permutation can be obtained by independently selecting a permutation within each category.  Algorithm 2 can also be generalized to allow for a caliper on the propensity score (or whatever continuous variable is used as the matching distance), although the changes are more involved.  For full details see Algorithm 3 in \S\ref{app:caliper} in the Supplementary materials.




\section{Simulations}


 To demonstrate the empirical value of match-adaptive inference for ensuring Type I error control in finite samples, we simulate from a data-generating process in which the null hypothesis is true and compare the rates of rejection for match-adaptive inference,  covariate-adaptive inference, and uniform inference after matching.
 We draw datasets of size $n=500$ with two observed covariates $X_1$ and $X_2$ sampled as independent mean-zero normal random variables, the first with variance 5  and the second with variance 1. We construct treatment status via one of two logistic models specified as follows:
 
 \vspace{-4em}
 
\begin{align*}
\text{logit}\left[\mathbb{P}(Z=1 \mid X) \right] &=  0.1 + 0.7\cdot X_1-0.4\cdot X_2 \quad \quad \quad &\text{(linear)}\\
\text{logit}\left[\mathbb{P}(Z=1 \mid X) \right] &=  0.2 + 0.7\cdot X_1-0.4\cdot X_2 + \log\left(|X_1|\right) - 0.5\cdot X_2^2 \quad \quad \quad &\text{(nonlinear)}
\end{align*}

 \vspace{-2em}

\noindent and outcomes using one of the two following models:

 \vspace{-4em}

\begin{align*}
Y &= X_1 + 2X_2 + \epsilon \quad &\text{(linear)}  \quad \quad \quad \quad 
Y&= 4|X_1|^3 + 6\sin(X_1)+ 2X_2 + \epsilon \quad &\text{(nonlinear)}
\end{align*}

 \vspace{-2em}

\noindent where in both models the $\epsilon$ terms are independent standard normal random variables.

On each dataset so generated we fit a propensity score using logistic regression on $X_1$ and $X_2$ and conducted matching using a robust Mahalanobis distance \citep[\S 9.3]{rosenbaum2020design}, either with or without a propensity score caliper (equal to 0.2 standard deviations of the estimated propensity scores).  When used, the caliper was strictly enforced so that if some treated units could not be matched, a minimal number of treated units was excluded in order to make the match feasible.  After matching, one-sided hypothesis tests were conducted using uniform inference, covariate-adaptive randomization inference, and match-adaptive randomization inference, using  both estimated and true propensity scores for the latter two methods. Note that match-adaptive inference always 
uses the estimated propensity scores for matching and in Algorithm 2.
 We examined both the difference-in-means statistic and the regression-adjusted test statistic of \citet{rosenbaum2002covariance} that uses ordinary linear regression (on $X_1$ and $X_2$) to residualize pair differences.  We conduct 21,600 replications for each unique combination of simulation parameters and reported Type I error as estimated by the proportion of rejections at the 0.05 level.

Table \ref{tab:simresults} summarizes the results of these simulations.  With the exception of certain cases in which the propensity score model is misspecified (under which we do not expect match-adaptive inference with estimated propensity scores to perform well), the match-adaptive procedure consistently controls Type I error at or near 0.05.  In contrast, covariate-adaptive inference over the full space of permutations frequently fails to control Type I error when calipers and regression adjustment are absent, and sometimes when one of these adjustments is present.

\section{Does HRT reduce heart disease risk?}
\label{sec:real_data}
\subsection{Study population and matched design}

We conduct a matched study using data from the Women's Health Initiative study (WHI).  WHI includes both a clinical trial and observational study. Eligible women were between the ages of 50 and 79 at baseline, postmenopausal, free from medical conditions predicting a survival of less than three years, and likely to live in the same region for at least three years. Women meeting these criteria and expressing interest in participating first underwent eligibility screening for the clinical trial. 
Those who did not qualify for the clinical trial or declined participation in it were invited to be part of the observational study 
\cite{women1998design}. 
Our WHI study cohort (obtained through the BioLINCC 
\noindent at the National Heart, Lung, and Blood Institute) consists of 53,045 women who participated in the WHI observational study, had a uterus, and were not using unopposed estrogen when they enrolled in the study. We adhered to the sample selection criteria outlined in \cite{prentice2005combined} and 
\cite{yu2021information}. In line with the existing literature, we define the treated group as current users of a combined estrogen-plus progestin preparation at baseline and the control group as individuals who have never used or are former users. 
The outcome of interest is an indicator for any occurrence of coronary heart disease (CHD), defined as the occurrence of clinical myocardial infarction (MI), definite silent MI or a death due to definite CHD or possible CHD.  


We begin by constructing an optimal propensity score match between patients who received HRT (the treated) and patients who did not (the controls) within 5-year age strata (50-54 years old, 55-59 years old, 60-64 years old, 65-69 years old, 70-74 years old, 75-79 years old). The matching step took 4.23 hours. The standardized mean differences for the demographic and clinical covariates are summarized in Table~\ref{tb:balance} in \S\ref{app:whi} in the Supplementary materials. The summary statistics of the propensity score before and after matching are summarized in Table~\ref{tb:ps} in \S\ref{app:whi} in the Supplementary materials. Results in these two tables show that our matching greatly improve the covariate balance. Although the standardized mean difference for the propensity score is slightly greater than 0.2, the summary statistics are very similar in the matched treated and control groups.

\subsection{Impact of HRT on CHD in matched pairs}
Motivated by findings of heterogeneous treatment effects across age of initiation of HRT in other datasets  \citep{lobo2017hormone, cagnacci2019controversial}, 
we plan our outcome analysis to incorporate both a test for an overall protective effect of HRT across the entire matched sample and a set of subgroup tests within patient categories.
Based on our matched sample, HRT decreases CHD risk by 0.64\% on average in the overall population.  Under match-adaptive randomization inference, the effect is significant at the 0.05 level but not at the 0.01 level (p-value$=0.035$) when testing is against the 
one-sided alternative hypothesis that HRT can decrease the occurrence of CHD. We also conducted inference using uniform randomization inference and covariate-adaptive randomization inference. The sample means of null draws for uniform randomization inference, covariate adaptive randomization inference and our new procedure are 0.000, -0.001, -0.004, respectively. The computation time is around 0.3 hours for the new inference method, longer than the two other forms of inference (each took less than 3 minutes), but an order of magnitude shorter than the time taken to construct the match.
Figure \ref{fig:null_distros} gives density plots of 50,000 null draws for each method. The corresponding one-sided p-values are 0.000, 0.002, 0.035, respectively.  The match-adaptive test views the nominally beneficial estimate of the treatment effect as much less surprising 
 under the null than do the uniform and covariate-adaptive tests, such that its result is not significant at the 0.01 level;   
 in Figure \ref{fig:null_distros}, the match-adaptive null distribution concentrates in the negative part of the real line with substantially more probability mass near the observed value of the test statistic. 
Intuitively, the match-adaptive test has identified the presence of pairs with large propensity score gaps and prevented them from being permuted.   These pairs tend also to have systematically better outcomes among the treated, since the propensity score is negatively correlated with the outcome, with correlation -0.11 in the full observational study.  
Clearly, ignoring $Z$-dependence risks anticonservative inference inappropriately permuting pairs that  must remain fixed in the true conditional distribution.

\begin{figure}[h!]
	\centering
	\includegraphics[scale=0.5]{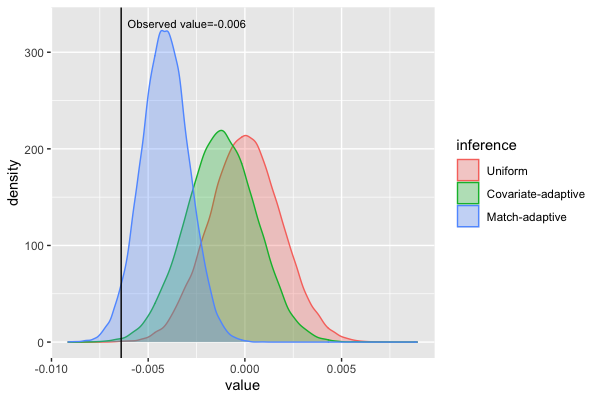}
	\caption{Null distribution for the difference-in-means estimator under uniform, covariate-adaptive, and match-adaptive inference.}
	\label{fig:null_distros}
\end{figure}

Next, in order to investigate the influence of HRT initiation age, we split the matched data into subgroups according to the age of the treated individual initiating HRT (before 50 years old, 50-59 years old, 60-69 years old, and at least 70 years old). For each subgroup, we estimate the treatment effect and calculate p-values based on match-adaptive inference. For reference, we also provide results based on uniform randomization inference and covariate-adaptive randomization inference.
Results in Table~\ref{tb:suboutcome} suggest significant beneficial effects of HRT on CHD for women who start HRT when 50-59 years old. At other age intervals, starting HRT has no significant effects reducing CHD occurrences. This conclusion is consistent with the current recommendation of initiating HRT for young and healthy women (aged 50-59 years old) \citep{cagnacci2019controversial,lobo2017hormone}.  


\begin{table}[ht]
	\centering
	\resizebox*{\textwidth}{!}{
		\begin{tabular}{rrrrrrrrr}
			\hline
			HRT Initiation Age&\multicolumn{2}{c}{$<50$ years old} & \multicolumn{2}{c}{50-59 years old} & \multicolumn{2}{c}{60-69 years old} & \multicolumn{2}{c}{70-79 years old}\\\hline
			\# of pairs &\multicolumn{2}{c}{4303}&\multicolumn{2}{c}{10498}&\multicolumn{2}{c}{2349}&\multicolumn{2}{c}{359} \\\hline
			Estimated effect&\multicolumn{2}{c}{-0.005}&\multicolumn{2}{c}{-0.006}&\multicolumn{2}{c}{-0.006}&\multicolumn{2}{c}{-0.025}\\\hline
			& null mean & p-value& null mean & p-value& null mean & p-value& null mean & p-value\\\hline
			Match-adaptive&-0.004&0.304&-0.004&0.045&-0.006&0.454&-0.003&0.107\\	
			Covariate-adaptive&-0.002&0.167&-0.001&0.012&0.000&0.170& 0.000&0.104\\
			Uniform &-0.000&0.060&0.000&0.003&-0.000&0.170&-0.000& 0.104\\\hline
	\end{tabular}}
\caption{Estimated treatment effects, mean of null distribution and p-values.}
\label{tb:suboutcome}
\end{table}

\section{Discussion}
\label{sec:discussion}

To the best of our knowledge, the newly proposed match-adaptive inference algorithm is the first procedure for matched observational studies providing explicit guarantees against $Z$-dependence, and our simulations and case study demonstrate its potential impact on empirical Type I error control.  We urge researchers to be aware of the potential for anticonservative testing when matches are not exact, and to use the machinery of covariate-adaptive and match-adaptive inference to assess the downstream impact of inexact matching.

Our findings about the impact of HRT on heart disease based on match-adaptive randomization inference agree qualitatively with the modern understanding that such effects depend on patients' initiation age.  Although we find benefits for women initiating therapy from age 50-59 when inference is conducted at the 0.05 level, and hence also reject the null hypothesis of no effect across the entire population, we fail to detect benefit for other age groups, and our inference procedures see more limited evidence for such benefits in the WHI data than  other testing methods with less comprehensive Type I error control.  Evidence about the presence or absence of an effect of hormone replacement therapy for women younger than 50 appears quite limited in the broader literature. In our study the nominal effect size for women under 50 is not much smaller than the effect size for the significant 50-59 group, and the match-adaptive randomization inference null distribution has a similar mean, but the sample size was much smaller.  Further study to interrogate possible beneficial effects for this age group seems warranted.

Several directions for future methodological work are also apparent.  First, a method of sensitivity analysis for unobserved confounding compatible with match-adaptive inference is needed. 
Although \citet{rosenbaum2002observational} introduced a comprehensive framework for sensitivity analysis in matched observational studies under uniform inference, 
adapting it to match-adaptive inference requires new technical tools.
Another direction is generalizing our algorithms to a broader class of matching methods.  While we focus on pair matching, 
when controls are prevalent or overlap is limited it can be helpful to extend the new inference framework to matching with multiple controls \citep{ming2000substantial} or full matching \citep{hansen2004full}.  
Extending our approach based on matching with distances on a univariate score  to  other common strategies such as matching on Mahalanobis distances \citep{rubin1980bias}  or under balance constraints \citep{pimentel2015large} could also greatly expand its practical applicability. 
Finally, there are interesting parallels between our contributions and the procedures of \citet{basse2019randomization} and \citet{puelz2022graph} for conducting randomization tests under interference, 
where 
the key is to identify the conditional distribution of an auxiliary event given the treatment vector. 
This is similar to our setting, in which the optimal matching partition 
 takes the place of 
 the auxiliary event. While the motivations for these procedures are different, the connection suggests a general strategy for adapting randomization tests effectively to complex designs or data generating processes.


\bibliographystyle{asa}
\bibliography{zdep_arxiv_v1}

\begin{landscape}
\begin{table}
\small
\begin{center}
\begin{tabular}{cc|cc|c|cc|cc}
\multicolumn{2}{c|}{Data-generating process} &  & &Uniform  &\multicolumn{2}{c|}{Covariate-adaptive inference} & \multicolumn{2}{c}{Match-adaptive inference} \\
$Y$-model & $Z$-model & Caliper? & Regression? & inference & Est. prop. score &  True prop. score &Est. prop. score & True prop. score \\
\hline
\multirow{4}{*}{Linear} & \multirow{4}{*}{Linear} &  & & \textbf{1.000} & \textbf{0.685} & \textbf{0.598} & {0.000} & {0.001} \\
& & X & & {0.005} & {0.003} & \textbf{0.078} & 0.034 & 0.053\\
& & & X & {0.003} & {0.007} & {0.008} & {0.001} & {0.001} \\
& & X & X & 0.033 & 0.032 & 0.034 & 0.064 & 0.064 \\\hline
\multirow{4}{*}{Nonlinear} & \multirow{4}{*}{Linear} &  & & \textbf{0.134} & {0.006} & {0.006} & {0.001} & {0.001} \\
& & X & & {0.053} & {0.050} & {0.065} & 0.048 & 0.053\\
& & & X & \textbf{0.337} & \textbf{0.366} & \textbf{0.371} & {0.001} & {0.001} \\
& & X & X & 0.046 & 0.48 & 0.049 & 0.049 & 0.053 \\\hline
\multirow{4}{*}{Linear} & \multirow{4}{*}{Nonlinear} &  & & \textbf{1.000} & \textbf{0.699} & \textbf{0.594} &{0.032} & {0.010} \\
& & X & & {0.021} & {0.015} & {0.062} & {0.013} & 0.056\\
& & & X & {0.006} & {0.011} & {0.018} & {0.029} & {0.025} \\
& & X & X & 0.038 & 0.039 & 0.043 & 0.046 & 0.048 \\\hline
\multirow{4}{*}{Nonlinear} & \multirow{4}{*}{Nonlinear} &  & & \textbf{0.987} & \textbf{0.885} & \textbf{0.486} & {0.049} & {0.019} \\
& & X & & \textbf{0.199} & \textbf{0.185} & \textbf{0.106} & \textbf{0.100} & 0.056\\
& & & X & {0.000} & {0.000} & {0.000} & {0.053} & {0.019} \\
& & X & X & \textbf{0.138} & \textbf{0.136} & {0.066} & \textbf{0.119} & 0.051\\ \hline\\
\end{tabular}
\end{center}
\caption{Results from the simulation study (21,600 replications, each dataset with 500 observations and two measured covariates).  The first four columns specify the simulation setting, and the final five columns report Type I error rates under different forms of inference: uniform inference (which does not use propensity scores), and covariate-adaptive and match-adaptive inference with estimated and true propensity scores used to calculate permutation probabilities 
{in the inference step}. Type I error estimates  are each based on 2,160 replications; and those significantly larger than 0.05 (under a Bonferroni correction scaled to the number of rates reported across the figure) are marked in bold. 
}
\label{tab:simresults}
\end{table}
\end{landscape}

\def\spacingset#1{\renewcommand{\baselinestretch}%
{#1}\small\normalsize} \spacingset{1}

\doublespacing

\setcounter{section}{0}
\renewcommand{\thesection}{\Alph{section}}
\def\thetable{S\arabic{table}}
\def\theequation{S\arabic{equation}}
\section{Additional details for the algorithm and case study}
\subsection{Details on optimal matching}
\label{app:matching}

{
In this section we give a detailed mathematical characterization of the optimization problem solved to produce the matched pairs encoded by $\mathcal{M}_{opt}$ as described in Section \ref{subsec:setup}.  For this section only, we use single-indexing of elements in $\mathbf{Z}$ and $\mathbf{X}$, e.g. $\mathbf{Z} = (Z_1, \ldots, Z_N)$  (in contrast to the double-indexing used to describe matched units elsewhere).  Let $K = \min\left\{\sum^N_{i=1}Z_i, N- \sum^N_{i=1}Z_i\right\}$ and let $d(X_i,X_j)$ be some nonnegative real-valued covariate distance defined for all $i,j \in 1, \ldots, N$.  
The investigator matches units into $K$ pairs by solving an optimization problem to minimize the average distance $d(\cdot,\cdot)$ between paired units.  Formally, letting {$\mathfrak{M}$} be the set of all collections of $K$ disjoint pairs $(i,i')$ with $i \neq i'$ and $i, i' \in \{1, \ldots N\}$, an optimal pairing is obtained by solving the following problem:
\begin{align}
\label{eqn:opt}
&\text{argmin}_{\mathcal{K} \in \mathfrak{K}}\sum_{(i,i') \in \mathcal{M}}d\left(X_{i}, X_{i'}\right) \\
&\text{s.t.}\nonumber\\
&Z_{i} + Z_{i'} = 1 \quad \text{ for all } (i,i') \in\mathcal{M}.\nonumber
\end{align}
We call the solution to this problem {$\mathcal{M}_{opt}$;  object $\mathcal{U}_{opt}$ collects the indices of the remaining unmatched units and their treatment indicators}.  

{
Note that the definition of {$\mathfrak{M}$}  given here requires that exactly $K$ pairs be formed.  In Section \ref{subsec:exact_caliper}, different versions of the optimization problem are discussed under which the total number of pairs may be less than $K$.  Under exact matching, a separate version of problem (\ref{eqn:opt}) is solved for each category of a nominal variable; since treated units may outnumber controls in some categories while controls outnumber treated units in others, the total number of pairs may be less than $K$.  Under caliper matching, pairs are only formed between units for which $d(X_i,X_j) \leq c$ for some fixed value $c$, which may make it impossible to form $K$ pairs.  Formally, define {$\mathfrak{M}^c$}  as the set of all collections of disjoint pairs $(i,i')$ for which $i \neq i'$, $i,i' \in \{1, \ldots, N\}$, and $d(X_i,X_{i'}) \leq c$, and for any {$\mathcal{M} \in \mathfrak{M}^c$} let $|\mathcal{M}|$ give the number of pairs in $\mathcal{M}$.  Then optimal matching with a caliper $c$ is conducted by solving the following optimization problem:
\begin{align}
\label{eqn:optcal}
&\text{argmin}_{{\mathcal{M} \in \mathfrak{M}^c}}\sum_{(i,i') \in \mathcal{M}}d\left(X_{i}, X_{i'}\right) \\
&\text{s.t.}\nonumber\\
&Z_{i} + Z_{i'} = 1 \quad \text{ for all } (i,i') \in\mathcal{M}.\nonumber\\
&|{\mathcal{M}}| = \max_{{\mathcal{M}}' \in {\mathfrak{M}}^c}|{\mathcal{M}}'|. \nonumber
\end{align}
For more discussion of exact matching and caliper matching as optimization problems, see \citet{rosenbaum2020modern}. 
}

\subsection{Subroutine definitions.}\label{app:subroutine}

\begin{flushleft}
\hspace{2em} \textbf{Subroutine} \emph{ConstructOverlapGraph} (input: list of paired propensity score values of form $(\lambda_1^T, \lambda_1^C), \ldots, (\lambda_K^T, \lambda_K^C)$, with treated unit's propensity score listed first):
\begin{enumerate}
\item Order all matched pairs by decreasing propensity score for the treated unit $\lambda^T_k$ and index the pairs by $k=1,\dots,K$. 
\item For $k = 1, \ldots, K$, do the following:
\begin{enumerate}[topsep=0.2pt, partopsep =0.2pt, itemsep=0.2pt, parsep = 0.2pt]
\item If $\lambda_k^T > \lambda_k^C$ and  $k < K$, then for $j = k+1, \ldots, K$ do the following:
\begin{enumerate}
\item Check if $\lambda_k^T > \lambda_j^T > \lambda_k^C$.  If so, pairs $k$ and $j$ overlap; if not, exit the (inner) loop.
\end{enumerate}
\item Otherwise (i.e., $\lambda_k^T < \lambda_k^C$), if $k > 1$, then for $j = k-1, \ldots, 1$ do the following:
\begin{enumerate}
\item Check if $\lambda_k^T < \lambda_j^T < \lambda_k^C$.  If so, pairs $k$ and $j$ overlap; if not, exit the (inner) loop.
\end{enumerate}
\end{enumerate}
\item Construct and return an undirected graph with node set $\mathcal{N} = \{1, \ldots, K\}$ and an edge set consisting of all pairs $(k,j)$ labeled overlapping.
\end{enumerate}
\end{flushleft}

\textbf{Subroutine} \emph{ConstructMetaComponents} (input: list of propensity scores for all units $\lambda_1, \ldots, \lambda_N$ and connected components of overlap graph $\mathcal{S}_1, \ldots, \mathcal{S}_r$).
\begin{enumerate}
\item Sort the connected components of the graph $\mathcal{S}_r, r=1,\dots, R$ by increasing maximum propensity, $\max\mathcal{Q}_r$.
\item Initiate a new index $i=0$. For each $r = 1, \ldots, R$:
\begin{enumerate}
	\item Let $j_*$ be the index of the unit with the largest propensity score $\lambda_*$ such that $\lambda_* < \min \mathcal{Q}_r$ (or $\lambda_*=-\infty$ if no such unit exists).
	\item If unit $j_*$ is an unmatched control or if $\lambda_*$ is not finite, increment the index $i$ by setting $i=i+1$, begin a new meta-component $\mathcal{T}_i = \{\}$, and define $s_i^\ell = \lambda_*$.
	\item Redefine $\mathcal{T}_i$ as $\mathcal{T}_i \cup\{ r\}$. 
\item Let $j^*$ be the index of the unit with the smallest propensity score $\lambda^*$ such that $\lambda^* > \max \mathcal{Q}_r$ (or $\lambda^*=\infty$ if no such unit exists).
	\item If unit $j^*$ is an unmatched control or if $\lambda^*$ is not finite, define $s_i^u = \lambda^*$.
\end{enumerate}
\item Define $I$ as the final value of $i$ and return meta-component information $\{\mathcal{T}_i, s_i^\ell, s_i^u\}_{i=1}^I$.
\end{enumerate}

\textbf{Subroutine} \emph{CheckComponentUp} (input: meta-component information $\mathcal{O} = \{\mathcal{T}_i, s_i^\ell, s_i^u\}_{i=1}^I$ of the type returned by subroutine \emph{ConstructMetaComponents},  index $j \in \{1, \ldots, |\mathcal{T}_i|+1\},$, new running cost $C \in \mathbb{R}$, old running cost $D \in \mathbb{R}$, propensity score $q_{max} \in \mathbb{R}$, and component-wise permutation indicators $\mathbf{w}\in \{0,1\}^{|\mathcal{T}_i|}$ ):
\begin{enumerate}
\item If $j = |\mathcal{T}_i|+1$, add $\mathbf{w}$ to list of compatible vectors for meta-component $i$ and return.
\item Letting $j'$ be the index of the $j$th component in $\mathcal{T}_i$, set $C' = C + \min \mathcal{Q}_{j'} - q_{max}$, set $D' = D + \max \mathcal{Q}_{j'} - \min \mathcal{Q}_{j'}$, and $q_{max}' = \max \mathcal{Q}_{j'}$.
\item If the lowest unit in the $j$th component in $\mathcal{T}_i$ is treated, do the following:
\begin{enumerate}
\item Set the $j$th coordinate of $\mathbf{w}$ to 1 and call \emph{CheckComponentUp}$(\mathcal{O},j+1,C,D, q_{max}, \mathbf{w})$.
\item If $C' \geq D'$, also set the $j$th coordinate of $\mathbf{w}$ to 0 and call \emph{CheckComponentUp}$(\mathcal{O},j+1, C',D',q_{max}',\mathbf{w})$.
\end{enumerate}
\item Otherwise:
\begin{enumerate}
\item Set the $j$th coordinate of $\mathbf{w}$ to 0 and call \emph{CheckComponentUp}$(\mathcal{O},j+1,C,D, q_{max}, \mathbf{w})$.

\item If $C' \geq D'$, also set the $j$th coordinate of $\mathbf{w}$ to 1 and call \emph{CheckComponentUp}$(\mathcal{O},j+1, C',D',q_{max}',\mathbf{w})$.

\end{enumerate}

\end{enumerate}

\textbf{Subroutine} \emph{CheckComponentDown}  (input: meta-component information $\mathcal{O} = \{\mathcal{T}_i, s_i^\ell, s_i^u\}_{i=1}^I$ of the type returned by subroutine \emph{ConstructMetaComponents},  index $j \in \{1, \ldots, |\mathcal{T}_i|+1\},$, new running cost $C \in \mathbb{R}$, old running cost $D \in \mathbb{R}$, propensity score $q_{min} \in \mathbb{R}$, and component-wise permutation indicators $\mathbf{w}\in \{0,1\}^{|\mathcal{T}_i|}$ ):
\begin{enumerate}
\item If $j = 0$, add $\mathbf{w}$ to list of compatible vectors for meta-component $i$ and return.
\item Letting $j'$ be the index of the $j$th component in $\mathcal{T}_i$, set $C' = C + q_{min} - \max \mathcal{Q}_{j'}$, set $D' = D + \max \mathcal{Q}_{j'} - \min \mathcal{Q}_{j'}$, and $q_{min}' = \min \mathcal{Q}_{j'}$.
\item If the highest unit in the $j$th component in $\mathcal{T}_i$ is treated, do the following:
\begin{enumerate}
	\item Set the $j$th coordinate of $\mathbf{w}$ to 1 and call \emph{CheckComponentDown}$(\mathcal{O},j-1,C,D, q_{min}, \mathbf{w})$.
\item If $C' \geq D'$, also set the $j$th coordinate of $\mathbf{w}$ to 0 and call \emph{CheckComponentDown}$(\mathcal{O},j-1, C',D',q_{min}',\mathbf{w})$.
\end{enumerate}
\item Otherwise:
\begin{enumerate}
\item Set the $j$th coordinate of $\mathbf{w}$ to 0 and call \emph{CheckComponentDown}$(\mathcal{O},j-1,C,D, q_{min}, \mathbf{w})$.

\item If $C' \geq D'$, also set the $j$th coordinate of $\mathbf{w}$ to 1 and call \emph{CheckComponentDown}$(\mathcal{O},j-1, C',D',q_{min}',\mathbf{w})$.
\end{enumerate}
\end{enumerate}

\subsection{Adjustments for settings with very large meta-components}\label{app:large}

{
An important practical issue with {Algorithm 2} is that it runs orders of magnitude more slowly for matches with at least one large (in the sense of having many distinct sub-components) meta-component.  This is because the number of recursive calls grows exponentially in the size of the meta-component.   As discussed in Section \ref{subsec:speedup}, when inference is conducted approximately by taking a large number of draws $N_{sim}$ from the permutation distribution, it is be inefficient to calculate the full support of the match-adaptive distribution of treatment indicators in a given meta-component when that support is much larger than $N_{sim}$; most of the recursive calls refer to treatment vectors that need never be considered.  Algorithm 2' below modifies Algorithm 2 to use a rejection sampling approach, first drawing candidate treatment vectors and limiting recursive exploration to this family of treatment vectors, repeating as needed until $N_{sim}$ valid draws are obtained
In practice, one can use the approach of Algorithm 2' for the large meta-components only and sample from the small meta-components using the approach of Algorithm 2 instead.}

{
A second opportunity to boost computational performance for large meta-components is early stopping.  Many of the recursive calls performed by Algorithm 2 in large meta-components form new matched pairs that skip over several interior components, paths which are unlikely to lead to better matches because of the large costs incurred.  In many practical situations, it may be possible to stop recursing early in these cases by confirming that there is no hope of bringing the new path's cost below the cost of the pairs eliminated from the original match; for example, if the gap between the new path cost and the old path cost exceeds the remaining distance between the current propensity score location and the end of the meta-component, no such hope exists. To improve computational performance, we propose a short-circuiting approach where we check the remaining space in the meta-component as an upper bound for how much further improvement the new path can offer and stopping early (accepting all downstream treatment allocations) when the upper bound does not suffice to improve the match.  These modifications are implemented in Algorithm 3 in Section \ref{app:caliper} below, along with adjustments for caliper matching.}

\begin{flushleft}
\hspace{2em} \textbf{Algorithm 2‘}
\begin{enumerate}
	\item Run {steps 1-2} of Algorithm 2 to construct connected components and meta-components.
	\item Draw $N_{sim}$ random treatment permutations that are compatible with Algorithm 1. 
	
	For each treatment permutation, record the permutation indicators as $\mathbf{w}$ and collect these vectors into a (multi)set $\mathcal{W}_{initial}$.
	
	 For each $i =1, \ldots, I$:
\begin{enumerate}
\item Define a subroutine: \textbf{CheckComponentUp}$( j, C, D, q_{max}, \mathbf{w})$ where $C,D, q_{max} \in \mathbb{R}, j \in \{1, \ldots, |\mathcal{T}_i|+1\},$ and $\mathbf{w} \in \{0,1\}^{|\mathcal{T}_i|}$:
\begin{enumerate}
\item If no member of $\mathcal{W}_{initial}$ shares the first $j$ elements of $\mathbf{w}$ for meta-component $i$, return.
\item If $j = |\mathcal{T}_i|+1$, return.
\item Letting $j'$ be the index of the $j$th component in $\mathcal{T}_i$, set $C' = C + \min \mathcal{Q}_{j'} - q_{max}$, set $D' = D + \max \mathcal{Q}_{j'} - \min \mathcal{Q}_{j'}$, and $q_{max}' = \max \mathcal{Q}_{j'}$.
\item If the lowest unit in the $j$th component in $\mathcal{T}_i$ is treated, do the following:
\begin{enumerate}
\item Set the $j$th coordinate of $\mathbf{w}$ to 1 and call \textbf{CheckComponentUp}$(j+1,C,D, q_{max}, \mathbf{w})$.
\item If $C' < D'$ remove all elements of ,$\mathbf{W}_{initial}$ that share the first $j$ elements of $\mathbf{w}$ in component $i$.  Otherwise,  set the $j$th coordinate of $\mathbf{w}$ to 0 and call \textbf{CheckComponentUp}$(j+1, C',D',q_{max}',\mathbf{w})$.
\end{enumerate}
\item Otherwise:
\begin{enumerate}
\item Set the $j$th coordinate of $\mathbf{w}$ to 0 and call \textbf{CheckComponentUp}$(j+1,C,D, q_{max}, \mathbf{w})$.

\item If $C' < D'$ remove all elements of ,$\mathbf{W}_{initial}$ that share the first $j$ elements of $\mathbf{w}$ in component $i$.  Otherwise,   set the $j$th coordinate of $\mathbf{w}$ to 1 and call \textbf{CheckComponentUp}$(j+1, C',D',q_{max}',\mathbf{w})$.

\end{enumerate}

\end{enumerate}
\item Run \textbf{CheckComponentUp}$(1,0,0,s_i^\ell, \mathbf{w})$  for $\mathbf{w}=(0,...,0)$ and $\mathbf{w}=(1,0,\ldots,0)$.  
\item Define another subroutine: \textbf{CheckComponentDown}$( j, C,D, q_{min}, \mathbf{w})$ where $C,D, q_{min} \in \mathbb{R}, j \in \{0, \ldots, |\mathcal{T}_i|\},$ and $\mathbf{w} \in \{0,1\}^{|\mathcal{T}_i|}$:
\begin{enumerate}
\item If $j = 0$, add $\mathbf{w}$ to list of compatible vectors for meta-component $i$ and return.
\item Letting $j'$ be the index of the $j$th component in $\mathcal{T}_i$, set $C' = C + q_{min} - \max \mathcal{Q}_{j'}$, set $D' = D + \max \mathcal{Q}_{j'} - \min \mathcal{Q}_{j'}$, and $q_{min}' = \min \mathcal{Q}_{j'}$.
\item If the highest unit in the $j$th component in $\mathcal{T}_i$ is treated, do the following:
\begin{enumerate}
	\item Set the $j$th coordinate of $\mathbf{w}$ to 1 and call \textbf{CheckComponentDown}$(j-1,C,D, q_{min}, \mathbf{w})$.
\item  If $C' < D'$ remove all elements of ,$\mathbf{W}_{initial}$ that share the first $j$ elements of $\mathbf{w}$ in component $i$.  Otherwise, set the $j$th coordinate of $\mathbf{w}$ to 0 and call \textbf{CheckComponentDown}$(j-1, C',D',q_{min}',\mathbf{w})$.
\end{enumerate}
\item Otherwise:
\begin{enumerate}
\item Set the $j$th coordinate of $\mathbf{w}$ to 0 and call \textbf{CheckComponentDown}$(j-1,C,D, q_{min}, \mathbf{w})$.

\item  If $C' < D'$ remove all elements of ,$\mathbf{W}_{initial}$ that share the first $j$ elements of $\mathbf{w}$ in component $i$.  Otherwise, set the $j$th coordinate of $\mathbf{w}$ to 1 and call \textbf{CheckComponentDown}$(j-1, C',D',q_{min}',\mathbf{w})$.
\end{enumerate}
\end{enumerate}
\item Run \textbf{CheckComponentDown}$(|\mathcal{T}_i|,0,0,s_i^u, \mathbf{w})$  for $\mathbf{w}=(0,...,0)$ and $\mathbf{w}=(1,0,\ldots,0)$.  
\item If $|\mathcal{W}_{initial}|$ < $N_{sim}$, return to the beginning of step 2, and draw new permutation vectors as in step 2.  If any are already members of $\mathcal{W}_{initial}$,  add the new copies to $\mathcal{W}_{initial}$ too.  Repeat 2(a)-2(d) on the remaining new permutations and add whichever are retained to $\mathcal{WL}_{initial}$.  Repeat until $N_{sim}$ valid permutations are obtained.

\end{enumerate}
\item For each $j = 1, \ldots, N_{sim}$ concatenate chosen treatments across the $I$ meta-components to create a single global treatment permutation vector.  The resulting $N_{sim}$ vectors are used as null draws of $\mathbf{Z}$.
\end{enumerate} 
\end{flushleft}

\subsection{Extension to caliper matching
}\label{app:caliper}

Algorithm 2 is not valid for propensity score matching under a propensity score caliper {(as specified in problem (\ref{eqn:optcal}) with $d(X_i,X_j)$ given by an absolute propensity distance)}.  In particular, it may reject treatment draws even if the objective can only be improved by introducing pairs that violate the caliper, since the caliper is not explicitly accounted for or checked when new pairs across connected components are considered by the recursive procedures.  However, it is possible to adjust the algorithm to do such checking.  

First of all, we consider the easier case where all treated units are matched within the propensity score caliper under the original treatment assignment. If there are no unmatched controls in this case, calipers do not introduce problems for Algorithm 1 (we can still rule out all the same treatment draws because the new pairs formed by resolving a crossing match always have strictly smaller pair distances than the original pairs and must therefore satisfy the caliper). Similarly, in the presence of unmatched controls, the caliper do not introduce any issues for reconfiguring pairs within a connected component when the most extreme observerations in the component are removed as is done by the recursive procedures in Algorithm 2, since there is always a way to do this that leads to strictly smaller pair distances than in the original match.  Thus it suffices to add checks for whether the new matches formed across the original connected components satisfy the caliper, and (if not) whether there is a way to reconfigure the pairs to avoid the caliper violation. 

Another challenge with using a caliper is that it may exclude some treated units that have no candidate controls within the caliper. In this case, we need to check for two distinct problems: (1) settings where more matched pairs can be formed (2) settings where bringing an unmatched treated unit into the match improves the objective without changing the number of pairs.  Adaptations to the recursive algorithm make this possible.

These adaptations, as well as the early stopping procedure discussed in Section \ref{app:large}, are implemented together in Algorithm 3.

\textbf{Algorithm 3}
\begin{enumerate}
	\item Run steps 1-3 of Algorithm 1 (i.e. create overlap graph and its connected components).
	\item Let $U$ indicate whether unmatched treated units are present.
	\item Sort the connected components of the graph $\mathcal{S}_k, k=1,\dots, K$ by increasing maximum propensity, $\max\mathcal{Q}_k$.
	\item (Construct the meta-components)  Initiate a new index $i=0$. For each $k = 1, \ldots, K$:
	\begin{enumerate}
		\item Let $j_*$ be the index of the unit with the largest propensity score $\lambda_*$ such that $\lambda_* < \min \mathcal{Q}_k$ (or $\lambda_*=-\infty$ if no such unit exists).
		\item If unit $j_*$ is unmatched or if $\lambda_*$ is not finite, increment the index $i$ by setting $i=i+1$, begin a new meta-component $\mathcal{T}_i = \{\}$, and define $s_i^\ell = \lambda_*$ and $H_i^\ell = \left\{\begin{array}{cc}Z_{j_*} & \text{ if $s_i^\ell$ is finite}\\ NA & \text{otherwise.}\end{array}\right.$
		\item Redefine $\mathcal{T}_i$ as $\mathcal{T}_i \cup\{ k\}$.  
		\item Let $j^*$ be the index of the unit with the smallest propensity score $\lambda^*$ such that $\lambda^* > \max \mathcal{Q}_k$ (or $\lambda^*=\infty$ if no such unit exists).
		\item If unit $j^*$ is  unmatched  or if $\lambda^*$ is not finite, define $s_i^u = \lambda^*$ and $H_i^u = \left\{\begin{array}{cc}Z_{j^*} & \text{ if $s_i^u$ is finite}\\ NA & \text{otherwise.}\end{array}\right.$
		\item If $U=1$ and both $\lambda_*$ and $\lambda^*$ are finite, run the caliper algorithm of \citet{ruzankin2020fast} on the units in component $k$ plus a unit with propensity score $\lambda_*$ having the same treatment status as the lowest unit in $\mathcal{Q}_k$ and a unit with propensity score $\lambda^*$ having the same treatment status as the largest unit in $\mathcal{Q}_k$.  {If all units can be matched under the caliper, mark component $k$ as bypassable and define $t_k^\ell$ and $t_k^u$ as the smallest (largest) propensity score for a unit in component $k$ with the opposite treatment status to the unit with propensity score $\lambda_*$ ($\lambda^*$).} 
	\end{enumerate}
	\item Set $I = i$.
	
	\item (Preprocess treatment assignments for each meta-component) For each $i =1, \ldots, I$:
	\begin{enumerate}
		\item Define a subroutine: \textbf{CheckComponentUp}$( j, \Delta, q_{max}, \mathbf{w})$ where $D, q_{max} \in \mathbb{R}, j \in \{1, \ldots, |\mathcal{T}_i|+1\},$ and $\mathbf{w} \in \{0,1\}^{|\mathcal{T}_i|}$:
		\begin{enumerate}
			\item (\emph{Check for end of meta-component}): if $j = |\mathcal{T}_i|+1$ do the following:
			\begin{enumerate}
				\item If any one of the following is true:
				\begin{itemize}
					\item Either $H_i^\ell = NA$ or $H_i^u = NA$;
					\item  $H_i^\ell =  H_i^u$;
					\item $s_i^u - q_{\max}$ exceeds the caliper;
				\end{itemize}
				then  add $\mathbf{w}$ to list of compatible vectors for meta-component $i$.
				\item Return.
			\end{enumerate}
			\item (\emph{Determine whether to bypass component}).  Let $Z_\ell$ be the treatment status of the lowest unit in the $j$th component in $\mathcal{T}_i$, and define bypass indicator $B$ as follows:
			\[
			B = \left\{\begin{array}{cc} 1 & \text{ if } Z_\ell(1- w_j )+ (1-Z_{\ell})w_j = H_i^\ell \\ 0 & \text{otherwise}  \end{array}\right.
			\]
			\item  Let $j'$ be the index of the $j$th component in $\mathcal{T}_i$.  If $B = 0$, do the following:
			\begin{enumerate}
				\item Let $\Delta' = \Delta + \min \mathcal{Q}_{j'} - q_{max} - \left( \max \mathcal{Q}_{j'} - \min \mathcal{Q}_{j'}\right) $.
				\item Let $q_{match} = \min \mathcal{Q}_{j'}$ and $q_{max}' =  \max \mathcal{Q}_{j'}$.
			\end{enumerate}
			If instead $B=1$ and component $j'$ is bypassable:
			\begin{enumerate}
				\item Let $\Delta' = \Delta + t_{j'}^u - q_{max}$.
				\item Let $q_{match} = t_{j'}^\ell$ and $q_{max}' =  t_{j'}^u$. 
			\end{enumerate}
			\item (\emph{Early stopping}) If any of the following is true:
			\begin{itemize}
				\item (\emph{Caliper violation}) B=1 and component $j'$ is not bypassable;
				\item  (\emph{Caliper violation}) $q_{match} - q_{\max}$ exceeds the caliper;
				\item (\emph{Current path too poor})  $\Delta' > \max \mathcal{Q}_{j_{\max}} - q_{max}'$ (where $j_{\max}$ is the index of the highest component in $\mathcal{T}_i$) and either at least one of $H_i^u$ and $H_i^\ell$ is NA or $H_i^u=H_i^\ell$ ;
			\end{itemize}
			then construct all vectors $\mathbf{w}' \in \{0,1\}^{|\mathcal{T}_i|}$ such that the first $j$ elements of $\mathbf{w}'$ are identical to the first $j$ elements of $\mathbf{w}$, add them all to the list of compatible vectors for meta-component $i$, and return.
							\item (\emph{Check for improved objective}) If $\Delta' < 0$, return.
			\item (\emph{Recurse}):
			\begin{enumerate}
				\item Call \textbf{CheckComponentUp}$(j+1, \Delta',q_{\max}',\mathbf{w})$.
				\item If $j <  |\mathcal{T}_i|$, switch the $(j+1)$th element of $\mathbf{w}$ to the opposite status and call \textbf{CheckComponentUp}$(j+1, \Delta',q_{max}',\mathbf{w})$ again.
			\end{enumerate}
		\end{enumerate}
		\item Let $\mathbf{w}^1=(1,0..,0)$ and $\mathbf{w}^0=(0,..,0)$.    Run \textbf{CheckComponentUp}$(1,0,s_i^\ell, \mathbf{w}^0)$ and \textbf{CheckComponentUp}$(1,0,s_i^\ell, \mathbf{w}^1)$.  After both functions complete, label the resulting set of compatible vectors $\mathcal{W}_{up}$.  
		
		\item Define another subroutine: \textbf{CheckComponentDown}$( j, \Delta, q_{min}, \mathbf{w})$ where $D, q_{min} \in \mathbb{R}, j \in \{1, \ldots, |\mathcal{T}_i|+1\},$ and $\mathbf{w} \in \{0,1\}^{|\mathcal{T}_i|}$:
		\begin{enumerate}
			\item (\emph{Check for end of meta-component}): if $j = 0$ do the following:
			\begin{enumerate}
				\item If any one of the following is true:
				\begin{itemize}
					\item Either $H_i^\ell = NA$ or $H_i^u = NA$;
					\item  $H_i^\ell =  H_i^u$;
					\item $q_{\min} - s_i^\ell$ exceeds the caliper;
				\end{itemize}
				then add $\mathbf{w}$ to list of compatible vectors for meta-component $i$.
				\item Return.
			\end{enumerate}
			\item (\emph{Determine whether to bypass component}).  Let $Z_u$ be the treatment status of the largest unit in the $j$th component in $\mathcal{T}_i$, and define bypass indicator $B$ as follows:
			\[
			B = \left\{\begin{array}{cc} 1 & \text{ if } Z_u(1- w_j )+ (1-Z_{u})w_j = H_i^u \\ 0 & \text{otherwise}  \end{array}\right.
			\]
			\item  Let $j'$ be the index of the $j$th component in $\mathcal{T}_i$.  If $B = 0$, do the following:
			\begin{enumerate}
				\item Let $\Delta' = \Delta + q_{\min} - \max \mathcal{Q}_{j'} - \left( \max \mathcal{Q}_{j'} - \min \mathcal{Q}_{j'}\right) $.
				\item Let $q_{match} = \max \mathcal{Q}_{j'}$ and $q_{\min}' =  \min \mathcal{Q}_{j'}$.
			\end{enumerate}
			If instead $B=1$ and component $j'$ is bypassable:
			\begin{enumerate}
				\item Let $\Delta' = \Delta + q_{\min} - t_{j'}^\ell$.
				\item Let $q_{match} = t_{j'}^u$ and $q_{min}' =  t_{j'}^\ell$.
			\end{enumerate}
			\item (\emph{Early stopping}) If any of the following is true:
			\begin{itemize}
				\item (\emph{Caliper violation}) B=1 and component $j'$ is not bypassable;
				\item  (\emph{Caliper violation}) $q_{\min} - q_{match}$ exceeds the caliper;
				\item (\emph{Current path too poor})  $\Delta' > q_{min}' -  \min \mathcal{Q}_{j_{\min}} $ (where $j_{\min}$ is the index of the lowest component in $\mathcal{T}_i$) and either at least one of $H_i^u$ and $H_i^\ell$ is NA or $H_i^u=H_i^\ell$;
			\end{itemize}
			then construct all vectors $\mathbf{w}' \in \{0,1\}^{|\mathcal{T}_i|}$ such that the last $j$ elements of $\mathbf{w}'$ are identical to the last $j$ elements of $\mathbf{w}$, add them all to the list of compatible vectors for meta-component $i$, and return.
				\item (\emph{Check for improved objective}) If $\Delta' < 0$, return.
			\item (\emph{Recurse}):
			\begin{enumerate}
				\item Call \textbf{CheckComponentDown}$(j-1, \Delta',q_{\min}',\mathbf{w})$.
				\item If $j >1$, switch the $(j-1)$th element of $\mathbf{w}$ to the opposite status and call \textbf{CheckComponentUp}$(j-1, \Delta',q_{min}',\mathbf{w})$ again.
			\end{enumerate}
			
		\end{enumerate}
		\item Let $\mathbf{w}^1=(0,..,0,1)$ and $\mathbf{w}^0=(0,..,0)$.    Run \textbf{CheckComponentDown}$(|\mathcal{T}_i|,0,s_i^u, \mathbf{w}^0)$ and \textbf{CheckComponentDown}$(|\mathcal{T}_i|,0,s_i^u,\mathbf{w}^1)$.  After both functions complete, label the resulting set of compatible vectors $\mathcal{W}_{down}$ and define $\mathcal{W} = \mathcal{W}_{up} \cup \mathcal{W}_{down}$.
	\end{enumerate}

	\item Draw $\mathbf{W}$ by sampling $\mathbf{w_i}$ in meta-component $i$  from $\mathcal{W}_i$.
	
	\item Choose $\mathbf{Z}$ by setting all subjects in each component to the corresponding treatment assignment in $\mathbf{W}$.

\end{enumerate} 

\subsection{Covariate balance checks for the WHI study}\label{app:whi}

See Tables \ref{tb:balance} and \ref{tb:ps}.

\begin{table}[h!]
	\singlespacing
	\centering
		\small
		\begin{tabular}{lrrr|rr}  \hline
			& \multicolumn{3}{c|}{Covariate means}&\multicolumn{2}{c}{Std. mean diffs.}\\
			& & Control & Control & After & Before\\
			& Treated & (after match) & (before match) & match &  match \\   \hline
			Sample size & 17509 & 17509& 35536 & 17509& 35536\\\hline
			Age & 60.78 & 60.83 & 64.71 & -0.01 &\textbf{ -0.56} \\   
			Maternal history of MI & 0.18 & 0.18 & 0.18 & -0.01 & -0.01 \\   
			Maternal history of MI unknown & 0.02 & 0.02 & 0.02 & -0.01 & -0.03 \\   
			Maternal history of MI missing & 0.63 & 0.64 & 0.65 & -0.01 & -0.02 \\   
			Paternal history of MI & 0.32 & 0.31 & 0.29 & 0.02 & 0.07 \\   
			Paternal history of MI unknown & 0.02 & 0.02 & 0.02 & -0.01 & -0.03 \\   
			Paternal history of MI missing & 0.58 & 0.59 & 0.60 & -0.01 & -0.05 \\   
			Oral contraceptive use& 0.53 & 0.49 & 0.35 & 0.09 & \textbf{0.39} \\   
			Former smoker & 0.46 & 0.44 & 0.41 & 0.05 & 0.10 \\   
			Current smoker & 0.05 & 0.06 & 0.07 & -0.05 & -0.07 \\   
			Smoking missing & 0.01 & 0.01 & 0.01 & -0.00 & -0.01 \\   
			Hypertension ever & 0.24 & 0.25 & 0.32 & -0.01 & -0.16 \\   
			Hypertension missing& 0.01 & 0.01 & 0.01 & 0.00 & -0.02 \\   
			High cholesterol ever & 0.11 & 0.11 & 0.15 & -0.02 & -0.12 \\   
			Cholesterol missing & 0.02 & 0.02 & 0.02 & 0.02 & -0.04 \\   
			Diabetes ever & 0.03 & 0.03 & 0.06 & -0.03 & -0.15 \\   
			Diabetes missing& 0.00 & 0.00 & 0.00 & 0.00 & -0.01 \\   
			Bilateral oophorectony & 0.00 & 0.01 & 0.01 & -0.01 & -0.02 \\   
			Bilateral oophorectony missing& 0.00 & 0.00 & 0.01 & 0.00 & -0.03 \\   
			BMI & 25.82 & 26.49 & 27.47 & -0.12 & \textbf{-0.30} \\   
			BMI missing& 0.01 & 0.01 & 0.01 & -0.01 & -0.03 \\   
			Years of education & 15.68 & 15.42 & 14.86 & 0.09 & \textbf{0.29} \\   
			Education  missing & 0.01 & 0.01 & 0.01 & -0.01 & -0.02 \\   
			Family income (\$1000s) & 64.4 & 58.7 & 49.4 & 0.16 & \textbf{0.42} \\   
			Family income missing & 0.04 & 0.04 & 0.05 & -0.02 & -0.06 \\   
			Family income unknown & 0.02 & 0.02 & 0.04 & -0.02 & -0.08 \\   
			Black & 0.03 & 0.05 & 0.08 & -0.08 & \textbf{-0.24} \\   
			Hispanic & 0.03 & 0.04 & 0.04 & -0.05 & -0.06 \\   
			Asian & 0.04 & 0.04 & 0.03 & 0.02 & 0.06 \\   
			Age at menopause & 50.49 & 50.49 & 50.19 & 0.00 & 0.06 \\   
			Physical function score & 86.41 & 85.61 & 81.40 & 0.04 & \textbf{0.27} \\   
			Physical function score missing & 0.01 & 0.01 & 0.02 & 0.00 & -0.06 \\   
			History of breast cancer & 0.01 & 0.02 & 0.08 & -0.07 & \textbf{-0.36} \\   
			History of breast cancer missing & 0.01 & 0.01 & 0.01 & -0.00 & -0.02 \\   
			Propensity score & 0.44 & 0.40 & 0.28 & \textbf{0.23} & \textbf{0.93} \\    
			\hline
		\end{tabular}
		\caption{Balance table: Covariate means in the treated group, matched control group, and full control group; standardized mean differences after an before matching. Standardized mean differences with absolute value greater than 0.2 are \textbf{bolded}.}
	\label{tb:balance}
\end{table}

\begin{table}[h!]
\singlespacing
	\small
	\centering
	\begin{tabular}{rrrrrrr}  \hline 
		& Min. & 1st Qu. & Median & Mean & 3rd Qu. & Max. \\   \hline
		Before Matching & 0.00 & 0.19 & 0.31 & 0.33 & 0.47 & 0.88 \\   
		All Treated & 0.01 & 0.31 & 0.44 & 0.44 & 0.56 & 0.85 \\   
		All Control & 0.00 & 0.15 & 0.26 & 0.28 & 0.39 & 0.88 \\ \hline
		After Matching & 0.01 & 0.30 & 0.41 & 0.42 & 0.53 & 0.88 \\   
		Matched Treated & 0.01 & 0.31 & 0.44 & 0.44 & 0.56 & 0.85 \\   
		Matched Control & 0.01 & 0.29 & 0.39 & 0.40 & 0.50 & 0.88 \\    \hline
	\end{tabular}
\caption{Propensity score summary statistics.}
\label{tb:ps}
\end{table}

\section{Proofs of main results}
\subsection{Supporting lemmas for proofs}

In this section we present several lemmas helpful for proving Propositions \ref{prop:condsupport}-\ref{prop:unmatched}.  All place restrictions on the set of matches that can be optimal in a given configuration of treated and control propensity scores.  Only the first two lemmas are needed to prove Proposition \ref{prop:condsupport}; we do not offer a proof of the first since the result and its proof are given in \citet{saevje2021inconsistency}.


\begin{lemma}[\citet{saevje2021inconsistency}]
	\label{lem:nocross}
	An optimal propensity score match contains no crossing matches.
\end{lemma}

\begin{lemma}
	Suppose that there exists non-overlapping subsets $\mathcal{P}_1, \mathcal{P}_2, \ldots, \mathcal{P}_L$ of the unit interval where each $\mathcal{P}_i$ is a closed interval and where each $\mathcal{P}_i$ contains identical numbers of treated and control individuals (i.e. the propensity scores of those individuals lie in $\mathcal{P}_i$).  Then an optimal propensity score match pairs units only within the same interval $\mathcal{P}_i$.
	\label{lem:partition}
\end{lemma} 
\begin{proof}
	Suppose not.  Then we can show a crossing match exists, which is a contradiction by Lemma \ref{lem:nocross}.  Specifically, if the statement is false there exist two intervals $\mathcal{P}_i$,$\mathcal{P}_j$ with $i\neq j$ such that a treated unit in $\mathcal{P}_i$ (call this unit $t_1$) is matched to a control in $\mathcal{P}_j$ (call this unit $c_1$).  
	
	WLOG suppose that the subsets $\mathcal{P}_l$ are organized in increasing order of the propensity score values they contain, and that $i < j$; in addition, let $n_l$ represent the number of treated units (or controls) in interval $\mathcal{P}_l$. 
	
	Setting aside the previously-mentioned matched pair between a treated unit in $\mathcal{P}_i$ and a control in $\mathcal{P}_j$ and focusing on the  region $\bigcup^L_{l=j}\mathcal{P}_l$, there remain $\sum^L_{l=j}n_l$ unmatched treated units in this region and only $\sum^L_{l=j}n_l - 1$ unmatched controls.  Therefore, at least one treated unit in $\bigcup^L_{l=j}\mathcal{P}_l$ (call this unit $t_2$ will need to match to a control in $\bigcup^{j-1}_{l=1}\mathcal{P}_l$ (call this unit $c_2$).  Then 
	\begin{align*}
		\max\left\{\lambda(X_{t_1}), \lambda(X_{c_2})\right\} \leq \max_{i: \lambda(X_i) \in \bigcup^{j-1}_{l=1}\mathcal{P}_l}\lambda(X_i) <   \min_{i: \lambda(X_i) \in \bigcup^{L}_{l=j}\mathcal{P}_l}\lambda(X_i) \leq  \min\left\{\lambda(X_{t_2}), \lambda(X_{c_1})\right\}.
	\end{align*} 
	Therefore the pairs $(t_1,c_1)$ and $(t_2,c_2)$ create a crossing match and the proof is complete. 
\end{proof}

The next three lemmas are used in the proof of Proposition \ref{prop:unmatched}.  To exclude an incompatible $\mathbf{Z}'$ draw, it is both necessary and sufficient to identify at least one matched pair configuration that achieves a better objective value than the original match.  The first lemma establishes that  is sufficient to consider all ways to add exactly one unmatched control to the design, removing exactly one originally-matched control.
Let $\mathfrak{M}(\mathbf{X}, \mathbf{Z})$ represent the set of all possible pair matches for covariates $\mathbf{X}$ and treatment $\mathbf{Z}$, and for any $\mathcal{M} \in \mathfrak{M}(\mathbf{X}, \mathbf{Z})$ let $f(\mathcal{M})$ be its objective value; let $\mathcal{M}^*$ be some optimal match under the original treatment vector $\mathbf{Z}$.  For any $\mathbf{Z}' \in \mathcal{M}_{opt}$, observe that $\mathcal{M}^*$ is still a feasible match, since it pairs each treated unit in $\mathbf{Z}'$ to a control. Let $\mathfrak{M}^q_{\mathcal{M}^*}(\mathcal{X}, \mathbf{Z}')$ be the set of all pair matches under $\mathbf{Z}'$ in which the set of controls chosen differs from the set of controls in $\mathcal{M}^*$ (also under $\mathbf{Z}'$) by no more than $q$ subjects. 

\begin{lemma}
\label{lem:onesuff}
If $min_{\mathcal{M} \in \mathfrak{M}_{\mathcal{M}^*}^0(\mathcal{X}, \mathbf{Z}')}f(\mathcal{M}) = f(\mathcal{M}^*)$ (i.e. we cannot improve the match using the same set of matched controls), we have:
\[
\min_{\mathcal{M} \in \mathfrak{M}(\mathcal{X}, \mathbf{Z}')}f(\mathcal{M}) < f(\mathcal{M}^*) 
\quad \quad \text{iff} \quad \quad
\min_{\mathcal{M} \in \mathfrak{M}^1_{\mathcal{M}^*}(\mathcal{X}, \mathbf{Z}')}f(\mathcal{M}) < f(\mathcal{M}^*).
\]
\end{lemma}

\begin{proof}
	{Since $\mathfrak{M}^1_{\mathcal{M}^*}(\mathcal{X}, \mathbf{Z}') \subseteq \mathfrak{M}(\mathcal{X}, \mathbf{Z}')$, we have $\min_{\mathcal{M} \in \mathfrak{M}(\mathcal{X}, \mathbf{Z}')}f(\mathcal{M}) \geq
	\min_{\mathcal{M} \in \mathfrak{M}_{\mathcal{M}^*}^1(\mathcal{X}, \mathbf{Z}')}f(\mathcal{M})$ 
	Therefore, it is clear that if $\min_{\mathcal{M} \in \mathfrak{M}^1_{\mathcal{M}^*}(\mathcal{X}, \mathbf{Z}')}f(\mathcal{M}) < f(\mathcal{M}^*)$, we have $
	\min_{\mathcal{M} \in \mathfrak{M}(\mathcal{X}, \mathbf{Z}')}f(\mathcal{M}) < f(\mathcal{M}^*)$.}
	
{Next, we show the other direction. Suppose $\mathcal{M}^*$ is not optimal under $\mathbf{Z}'$ and denote the optimal match as $\mathcal{M}' \in  \mathfrak{M}(\mathcal{X}, \mathbf{Z}')$ (so $f(\mathcal{M}') < f(\mathcal{M}^*)$). Then we create a graph in which each treated unit is connected to a given control unig if it is matched to that control in either $\mathcal{M}'$ or $\mathcal{M}^*$. Since the two matches $\mathcal{M}^*$ and $\mathcal{M}'$ are not identical, there exists at least 1 connected part with more than 2 units. All connected parts with more than 2 units in this graph are chains of alternating treated and control units with controls on both ends (so the number of edges in each connected part is necessarily even and the number of controls in each connected part is the number of treated units in the same connected part plus 1). 
 Numbering the edges consecutively from either end, choosing only the odd edges gives the pair configuration for these units in one match and choosing only the even edges gives the pair configuration in the other.  For every such connected part, the objective value within it must be smaller in $\mathcal{M}'$ than in $\mathcal{M}^*$, otherwise $\mathcal{M}'$ would not be optimal under $\mathbf{Z}'$ since we could improve it by updating the match using the corresponding part in $\mathcal{M}^*$.  Therefore if we take match $\mathcal{M}^*$ and reverse only one of the connected parts with more than 2 units, we have improved the objective value while introducing only one new unmatched control, and the proof is complete.} 
\end{proof}

With Lemma \ref{lem:onesuff} established, the following lemma, proved originally by \citet{ruzankin2020fast}, simplifies our handling of large connected components by allowing us to restrict our attention to the largest and smallest propensity scores in each connected component when computing changes to objective value.

\begin{lemma}[\cite{ruzankin2020fast}]
	Suppose the data has the same number of treated and controls with propensity scores $\lambda_{l}^T$ and $\lambda_{l}^C$ respectively, $l=1,\dots,L$, such that $$\lambda_{1}^T\leq\lambda_{2}^T\leq\cdots\leq\lambda_{L}^T \quad \text{ and } \quad \lambda_{1}^C\leq\lambda_{2}^C\leq\cdots\leq\lambda_{L}^C.$$
	Then an optimal matching can be obtained by pairing the treated unit with propensity score $\lambda_{l}^T$ to the control with propensity score $\lambda_{l}^C$ for all $l=1,\dots,L$.
	\label{lem:optimal}
\end{lemma}

\subsection{Proof of Proposition \ref{prop:condsupport}}\label{app:proof1}

\begin{center}
\noindent\fbox{%
\vspace{2mm}
\parbox{\textwidth}{%
\vspace{2mm}
\textcolor{black}{Suppose that the initial match left no unmatched units.  Then for any compatible permutation $\mathbf{Z}'$ of $\mathbf{Z}$ generated by the algorithm, $\mathbb{P}( {\mathcal{M}_{opt}, {\mathcal{U}_{opt}}}\mid \mathbf{Z}', \mathcal{F}) =1$, and for any permutation $\mathbf{Z}'$ rejected, $\mathbb{P}( {\mathcal{M}_{opt}, {\mathcal{U}_{opt}}}\mid \mathbf{Z}', \mathcal{F}) =0$. }
}
}
\end{center}
\begin{proof}

First note that it suffices to determine whether the original match based on $\mathbf{Z}$ encoded by $\mathcal{M}_{opt}$ is still an optimal propensity score match or whether under $\mathbf{Z}'$ there exists some other pairing with a strictly lower objective value.

Represent the set of units in each connected component $r$ identified by Algorithm 1 as $\mathcal{S}_r$, and let $\mathcal{Q}_r = \left[\min_{i: i \in \mathcal{S}_r}\lambda(X_i), \max_{i: i \in \mathcal{S}_r}\lambda(X_i)\right]$.  By construction of the connected components, the $\mathcal{Q}_r$ are all non-overlapping closed intervals.  Under any permutation of $\mathbf{Z}$ {within pairs} $\mathcal{M}_{opt}$, the number of treated units in each connected component remains fixed and equal to the number of control units; therefore by Lemma~\ref{lem:partition}, any optimal propensity score matching must pair individuals only within the same connected component.

Now consider the subvectors of $\mathbf{Z}'$ and $\mathbf{Z}$ in a given connected component $\mathcal{S}_r$; call them $\mathbf{Z}_r'$ and $\mathbf{Z}_r$ respectively.  The previous paragraph tells us that the original match is still optimal if and only if within each connected component $k$, the treatment assignment $\mathbf{Z}_r'$ does not lead to a better within-component match than treatment assignment $\mathbf{Z}_r$.  To determine whether this is true for a given $\mathbf{Z}_r'$, notice that there are three possible cases:
\begin{enumerate}
	\item $\mathbf{Z}_r' = \mathbf{Z}_r$ (no treatment assignments in $\mathbf{Z}_r$ are switched by $\mathbf{Z}_r'$).
	\item  $\mathbf{Z}_r' = \mathbf{1}- \mathbf{Z}_r$ (all treatment assignments in $\mathbf{Z}_r$ are switched by $\mathbf{Z}_r'$).
	\item $\mathbf{Z}_r' \neq \mathbf{Z}_r$, $\mathbf{Z}_r'\neq \mathbf{1} - \mathbf{Z}_r$ (some but not all treatment assignments in $\mathbf{Z}_r$ are switched by $\mathbf{Z}_r'$).
\end{enumerate}
{Suppose first that case 3 holds for any $r$.  
Sort the pairs in $\mathcal{S}_r$ by the propensity score of the treated unit. Since not all treatment assignments are switched, there exits $i$ such that pair $i$ has switched treatment assignment and it has an adjacent pair $j$ with unswitched treatment assignment. Then these two pairs form a crossing match and therefore the match as a whole cannot be an optimal propensity score match.}

Now suppose that in every connected component $\mathcal{S}_r$, either case 1 or case 2 holds for subvectors $\mathbf{Z}_r', \mathbf{Z}_r$. Within each component, the match is still optimal because the solution to the optimal propensity matching is invariant to which group is labeled treated and  which labeled control when groups are equal in size.  Furthermore, by Lemma~\ref{lem:partition}, since the match within each connected component is optimal, the overall match is optimal.

To complete the proof, it suffices to observe that the set of permutations allowed by Algorithm 1 is exactly the set for which either case 1 or case 2 holds in each connected component.

\end{proof}


\subsection{Proof of Proposition~\ref{prop:unmatched}} \label{app:proof2}
\begin{center}
\noindent\fbox{%
\vspace{2mm}
\parbox{\textwidth}{%
\vspace{2mm}
For any permutation $\mathbf{Z}'$ of $\mathbf{Z}$ rejected by Algorithm 2, $\mathbb{P}( {\mathcal{M}_{opt}, {\mathcal{U}_{opt}}} \mid \mathcal{F},  \mathbf{Z}') = 0$, and for any permutation $\mathbf{Z}'$ not rejected, $\mathbb{P} ({\mathcal{M}_{opt}, {\mathcal{U}_{opt}}}\mid \mathcal{F},  \mathbf{Z}') =1$.}
}
\end{center}


\begin{proof}
	First note that since Algorithm 2 accepts only $\mathbf{Z}'$ draws also accepted by Algorithm 1, by Proposition \ref{prop:condsupport} there exists no better pairing among the units included in the original match.  In addition, by Lemma \ref{lem:onesuff}, if a better pairing exists using multiple unmatched controls, one also exists using only one unmatched control. Therefore, to establish the result it is sufficient to show that for any $\mathbf{Z}'$ accepted by Algorithm 2,  no better match among those incorporating exactly one of the originally-unmatched controls exists, and that for any $\mathbf{Z}'$ not accepted, some better match  incorporating exactly one of the originally-unmatched controls exists.

Suppose that one such originally-unmatched control $c^*$ now enters the match, improving the objective function.  Without loss of generality, let the new match be the best of all possible such one-step improvements (i.e. a minimizer of objective function $f(\cdot)$ on $\mathfrak{M}_{m^*}^1(\mathcal{X}, \mathbf{Z}')$ in the language of  Lemma \ref{lem:onesuff}).  
 $c^*$ must be adjacent to the meta-component containing the treated unit $t$ to which it is now paired, i.e. no other originally-unmatched control can lie between $c^*$ and that meta-component; otherwise we could improve the match by pairing $t$ to this other control instead. 
 
 Take any such $c^*$ as given.  We define the \emph{chain} associated with the new match incorporating $c^*$ (and the particular value $\mathbf{Z}'$ under which it is constructed) as a sequence of propensity scores $(\lambda_1, \lambda_2, \ldots, \lambda_{2K+1})$ such that each odd entry is the propensity score for a distinct control unit under $\mathbf{Z}'$, each even entry is the propensity score for a distinct treated unit under $\mathbf{Z}'$,  
  $\lambda_1 = \lambda_{c^*}$, $(\lambda_{2i}, 
  \lambda_{2i+1})$ are paired under the original match for all $i \in 1, \ldots K$, and $(\lambda_{2i-1}, \lambda_{2i})$ are paired under the new match for all $i \in 1, \ldots K$.  We define a sub-chain as any chain that can be obtained by taking only the first $2K' +1$ entries of the chain with $K' \leq K$; any larger chain for which the current chain is a subchain is referred to as a superchain.  We also define the value of a chain as:
  \[
  \sum^K_{i=1}|\lambda_{2i+1} - \lambda_{2i}| -   \sum^K_{i=1}|\lambda_{2i-1} - \lambda_{2i}| 
  \]
 Note that possible new matches incorporating $c^*$ are bijectively associated with chains, and that the value of the chain encodes the change in objective value associated with choosing the new match over the original match.  As such, to check if $\mathbb{P}( {\mathcal{M}_{opt}, {\mathcal{U}_{opt}}} \mid \mathcal{F}, \mathbf{Z}') = 1$, it suffices  to show that the minimum value over all possible chains is nonnegative, and to show $\mathbb{P}( {\mathcal{M}_{opt}, {\mathcal{U}_{opt}}} \mid \mathcal{F}, \mathbf{Z}') = 0$ it suffices to find one chain with negative value.  
 
Next, we prove that whenever a chain with negative value exists, one must exist that changes pairs only within a single meta-component.  Consider any negative-value chain $\theta$ that changes pairs within multiple meta-components.  Let $\theta_s$ be the largest sub-chain of $\theta$ that changes pairs only within a single meta-component $i_1$, so that the last control unit in $\theta_s$ is matched to a treated unit in a different meta-component $i_2$ in $\theta$.  If $\theta_s$ has negative value, the claim holds.  If instead $\theta_s$ has nonnegative value, we can create a new chain $\theta_t$ with (negative) value no greater than that of $\theta$ by beginning with an unmatched control nearest to meta-component $i_2$ and then following the same sequence of units that followed subchain $\theta_s$ in chain $\theta$.  If $\theta_t$ still incorporates multiple meta-components, the process can be repeated until a negative-value chain is produced that is contained within a single meta-component.  Therefore we consider only chains within a single meta-component for the remainder of the proof.

Now consider Algorithm 2.  The workhorse functions are CheckComponentUp and CheckComponentDown.  These functions each operate on a particular component $j$ in a particular meta-component $i$ for a particular vector of treatment assignments $\mathbf{w}$ in that meta-component (encoded as component-level treatments), and take as inputs two arguments $C$ and $D$ giving respective running costs for a new match and the original match.  As we will show, underlying any call to either function is an implicit chain $\theta_{j-1}$, which has value $C-D$, and which each function seeks to extend into a larger chain (by adding units from the component under consideration) so as to minimize the value of the final chain.  To streamline our analysis of this process, we define some additional quantities.  We let the treatment subvector $\mathbf{w}(j) \in \{0,1\}^j$ encode the treatments for the components numbered 1 through $j$ in meta-component $i$, and we let the completion set $\mathcal{C}(\mathbf{w}(j))$ be the set of all $\mathbf{w} \in \{0,1\}^{|\mathcal{T}_i|}$ whose first $j$ values are equal to the corresponding values in $\mathbf{w}_j$.  Let the index of $\theta_{j-1}$'s final control unit be $c^*_{j-1}$ (with associated propensity score $\lambda_{c^*_{j-1}})$.  Let the completion set  $\mathcal{C}(\theta_{j-1}, \mathbf{w})$ be the set of all superchains of $\theta_j$ incorporating only units within meta-component $i$ under treatment assignment $\mathbf{w}$.   We also let the $(j-1)$-restricted completion set $\mathcal{C}_{j-1}(\theta_{j-1}, \mathbf{w})$ be the subset of $\mathcal{C}(\theta_{j-1}, \mathbf{w})$ that is the union of $\{\theta_{j-1}\}$ and the set of all superchains of $\theta_{j-1}$ such that the treated unit immediately following $c^*_{j-1}$ does not lie in any of components $1, \ldots, j-1$.

We prove the theorem largely by an inductive argument on calls to CheckComponentUp, which consider for each meta-component the nearest unmatched control $c^*$ with propensity score smaller than those of units in the meta-component (an almost-identical argument, which we omit in the interest of concision, is also needed for CheckComponentDown to consider the nearest unmatched control with larger propensity score for each meta-component).  In particular, consider the following inductive hypothesis $A_j(\mathbf{w})$ for $j = 0, \ldots, |\mathcal{T}_i|+1$:

\begin{adjustwidth}{2cm}{}
$A_j(\mathbf{w})$: For any call to CheckComponentUp with component index $j$ and input treatment vector $\mathbf{w}'$ such that $\mathbf{w}'(j-1) = \mathbf{w}(j-1)$ , the following statements are true:
\begin{enumerate}
\item If $j = |\mathcal{T}_i| + 1$, no negative-value chain beginning with unit $c^*$ is possible when its units are assigned treatments according to $\mathbf{w}$.  
\item If $j \leq |\mathcal{T}_i|$ but the function exits without recursing further on treatment subvectors of $\mathbf{w}$, then a negative-value chain exists whenever its units are  assigned treatments according to $\mathbf{w}$.
\item For any new call to CheckComponentUp this instance of CheckComponentUp makes with treatment subvector $\mathbf{w}(j)$ and chain $\theta_j$, the value of $\theta_j$ is non-negative and either:
\begin{enumerate}[label=\alph*.]
\item   
the minimal value over chains in $\mathcal{C}_j(\theta_{j}, \mathbf{w})$ is no greater than the minimal value over chains in $\mathcal{C}(\theta_{j-1}, \mathbf{w})$, or
\item the minimal value over chains in $\mathcal{C}(\theta_{j-1}, \mathbf{w})$ is nonnegative.
\end{enumerate}
\end{enumerate}
\end{adjustwidth}

We start by proving the trivial base case $A_0(\mathbf{w})$.  For $j=0$, case 1 can be ignored and since we have not yet made a call to CheckComponentUp we can also ignore case 2; case 3 holds similarly because the recursive call is simply the initial call to CheckComponentUp and both $C(\theta_{j-1},\mathbf{w})$ and $C_j(\theta_j, \mathbf{w})$ refer to the set of all chains beginning with $c^*$ under this $\mathbf{w}$.

Now we will prove $A_j(\mathbf{w})$ assuming $A_0(\mathbf{w}), A_1(\mathbf{w}), \ldots, A_{j-1}(\mathbf{w})$ all hold.  First consider case 1, under which  $\mathbf{w}(j-1) = \mathbf{w}$ and $\mathcal{C}_j(\theta_{j},\mathbf{w})$ is the set consisting of $\theta_{j-1}$ alone.   From $A_{j-1}(\mathbf{w})$ we know that $\theta_{j-1}$ has nonnegative value, so we also know that the minimal value over chains in $\mathcal{C}(\theta_{j-1}, \mathbf{w})$ is nonnegative (since the minimal value over chains in $\mathcal{C}_j(\theta_{j}, \mathbf{w})$ is simply the value of $\theta_{j-1}$).
By considering the earlier inductive hypotheses in turn, we know also that $\theta_{j-1}$'s value is no greater than the minimal value across any of the sets  $\mathcal{C}(\theta_{j-2}, \mathbf{w})$, \ldots,  $\mathcal{C}(\theta_{0}, \mathbf{w})$, or that these sets themselves have nonnegative minimal value.  The last set is simply the set of all chains beginning with $c^*$.  This establishes case 1.
 
 Next we consider case 3.    It must be the case that 
 all overlapping pairs in the same connected component $\mathcal{S}_r$ share the same relative ordering for the treated and control propensity scores (denoted as $\lambda_k^T$ and $\lambda_k^C$ respectively). 
Specifically, either
$\lambda_k^T>\lambda_k^C$ holds for all $k\in\mathcal{S}_r$ or
$\lambda_k^T<\lambda_k^C$ is true for all $k\in\mathcal{S}_r$. Any deviation from this pattern leads to the existence of crossing matches, contradicting the optimality of the original match \citep[Lemma~3]{saevje2021inconsistency}.

 When treated units are larger than their matched counterparts in component $j$, CheckComponentUp implicitly sets $\theta_j = \theta_{j-1}$ (so that the input values $C$ and $D$ are unchanged in the recursive call).  Consider any other superchain $\theta'$ of $\theta_{j-1}$ consisting entirely of units in the first $j$ components.  By inductive hypothesis $A_{j-1}(\mathbf{w})$ we know that to identify a negative-value superchain (if it exists) it is sufficient to consider $\theta' \in \mathcal{C}_{j-1}(\theta_{j-1}, \mathbf{w})$, so when $j > 1$ we can ignore superchains that begin by matching $c^*_{j-1}$ to any treated unit in components $1$ through $j-1$ 
 Furthermore, if $j > 1$, $\theta'$ first extends $\theta_{j-1}$ by matching $c^*_{j-1}$ to any treated unit $t$ in $j$, and if at any later point in the extended chain
 some control $c$ in component $j$ is matched to a treated unit in one of components $1, \ldots, j-1$, a crossing match is created with the pair $(c^*_{j-1},t)$ and neither this chain nor any of its superchains can achieve the minimal value over $\mathcal{C}_{j-1}(\theta_{j-1}, \mathbf{w})$, since a strictly better match (corresponding to a different superchain) can be constructed instead.  Therefore, any value $\theta'$  that can achieve the minimal value must extend beyond $\theta_{j-1}$ only by changing pairs only within component $j$, and must therefore end with some control $c'$ in component $j$.  Numbering the treated units in component $j$ as $t_1, \ldots, t_{n_j}$ by increasing propensity score and the control units in component $j$ as $c_1, \ldots, c_{n_j}$ with $c_k = c'$, we can obtain a lower bound on the difference between the values of $\theta_{j-1}$ and any $\theta' \neq \theta_{j-1}$ meeting the minimization condition using  Lemmas \ref{lem:partition} and \ref{lem:optimal} as follows:
 \begin{align*}
 \left[\lambda_{t_1} - \lambda_{c^*_{j-1}} +  \sum^{k}_{i=1}(\lambda_{t_{i+1}} - \lambda_{c_i}) + \sum^{n_j}_{k + 1}(\lambda_{t_{i}} - \lambda_{c_i})\right] - \sum^{n_j}_{i=1}&(\lambda_{t_i} - \lambda_{c_i})  &= \lambda_{c_k} - \lambda_{c^*_{j-1}} \geq 0.
 \end{align*}
 Now consider any superchain of such a $\theta'$ ending in $c_k$.  Since the value of $\theta'$ exceeds the value of any $\theta_{j-1}$ by $\lambda_{c_k} - \lambda_{c^*_{j-1}}$, we can produce a chain with equal or lesser value by extending $\theta_{j-1}$ using the same sequence of units.  Therefore, any chain value that can be achieved by a chain in $\mathcal{C}(\theta_{j-1}, \mathbf{w})$ can also be achieved by a chain in $\mathcal{C}_j(\theta_j, \mathbf{w})$ with $\theta_j = \theta_{j-1}$.

 When instead treated units are paired to controls with larger propensity score values in component $j$, CheckComponentUp 
first extends  $\theta_{j-1}$  by  matching $c^*_{j-1}$ to the treated unit $t_{\min}(j)$ with smallest propensity score in component $j$, and the other matches are reconfigured so that the largest control unit in this component $c_{\max}(j)$ is left unmatched to any treated unit in the component (i.e. it is the last unit in the chain).  To determine the change in chain value driven by this extension, we number the treated units in component $j$ as $t_1, \ldots, t_{n_j}$ by increasing propensity score and the control units in component $j$ as $c_1, \ldots, c_{n_j}$ by increasing propensity score; by Lemmas \ref{lem:partition} and \ref{lem:optimal}  we obtain the following value for the change:
\begin{align*}
&= \sum^{n_j}_{i=1}(\lambda_{ci} - \lambda_{ti}) - \left[ \sum^{n_j}_{i=2}(\lambda_{c_{i-1}} - \lambda_{t_i} )  + (\lambda_{t_1} - \lambda_{c^*_{j-1}})\right]\\
&= \lambda_{c_{n_j}} - \lambda_{t1} - (\lambda_{t1} - \lambda_{c^*_{j-1}}).
\end{align*}
In the first line we use the fact that $\lambda_{c(i-1)}  > \lambda_{ti}$ (otherwise we can show that not all units can be part of the same connected component $j$). Compare this chain $\theta_{j}$ to any other superchain $\theta'$ of $\theta_{j-1}$ consisting entirely of units in components 1 through $j$.  Hypothesis $A_{j-1}(\mathbf{w})$ is true, so either all superchains of $\theta_{j-1}$ under $\mathbf{w}$ have nonnegative value (which suffices for case 3) or it to find the minimum value of all superchains sufficient to consider members of $\mathcal{C}_{j-1}(\theta_{j-1}, \mathbf{w})$.  First suppose
  $\theta' \neq \theta_{j-1}$; then $c^*_{j-1}$ must be matched to $t_1$ if $\theta'$ achieves the minimal value among superchains of $\theta_{j-1}$ otherwise a crossing match is formed between the pairs involving $c^*_{j-1}$ and $t_1$.   If $j > 1$ and some control unit in component $j$ is matched to a treated unit in components 1 through $j-1$ as part of $\theta'$, another crossing match is formed between this pair and the pair $(c^*_{j-1}, t_1)$, so it is suffiicent to consider superchains $\theta'$ that use only units within component $j$.  Finally, note that if any control $c_k$ besides $c_{n_j}$ is the last control in $\theta'$, then the new match associated with this chain cannot be optimal, since it leaves an unmatched control in the middle of at least one treated-control pair (otherwise not all units in component $j$ can be part of the same connected component), and any further extensions of $\theta'$ create a crossing match with the pair involving $c_{n_j}$.
If alternatively,  $\theta' = \theta_{j-1}$, i.e. the current chain is not extended using any units in this component, then either nonnegative-valued $\theta_{j-1}$ is itself the minimal-value chain  in $\mathcal{C}(\theta_{j-1}, \mathbf{w})$, or there exists some chain $\theta''$ minimizing value over $\mathcal{C}(\theta_{j-1}, \mathbf{w})$ that extends $\theta_{j-1}$ by first matching to a treated unit $t$ outside components $1, \ldots, j$ and achieves a lower value than any chain in $C_j(\theta_j, \mathbf{w}$.  But this is a contradiction, since we can construct a chain in  $C_j(\theta_j, \mathbf{w})$ with no greater value than $\theta''$ by replacing the pair $(c^*_{j-1},t)$ with the sequence by which $\theta_j$ extends $\theta_{j-1}$, and the pair $(c_{n_j}, t)$ (by the argument given above about why optimal chains do not match to treated units in earlier components, we know these units are not already part of $\theta''$). The difference in chain value between $\theta''$ and this new chain is equal to:
\begin{align*}
\lambda_{t} &- \lambda_{c^*_{j-1}} - \left[\lambda_t - \lambda_{c_{n_j}} + \lambda_{c_{n_j}} - \lambda_{t1} - (\lambda_{t1} - \lambda_{c^*_{j-1}}) \right] = 2\lambda_{t1} - 2\lambda_{c^*_{j-1}} \leq 0
\end{align*}
Therefore, either $\mathcal{C}(\theta_{j-1},\mathbf{w})$ has non-negative minimal value or it has minimal value no greater than the minimal value of $\mathcal{C}_j(\theta_j, \mathbf{w})$ and case 3 is proven.

Finally, consider case 2.  If the function exits without recursing further on subvectors of $\mathbf{w}$, this can only be because component $j$ has treated units smaller than control units in each matched pair, and the extended $\theta_j$ constructed as discussed in the previous paragraph achieves a negative value, so the condition holds.

Therefore, by induction, for any $\mathbf{w}$, $A_j(\mathbf{w})$ holds for all $j = 1, \ldots, |\mathcal{T}_i|+1$.  A similar set of inductive hypotheses can be shown for the function CheckComponentDown, which checks all chains for a meta-component beginning with the nearest unmatched control with larger propensity score.
 By our reasoning above and the bijection between negative chains and alternative matches with improved objective functions within a meta-component, this means that the treatment vectors approved by both CheckComponentDown and CheckComponentUp are exactly the meta-component treatment assignments for which no better within-meta-component match exists. Therefore Algorithm 2 allows exactly the set of full-dataset treatment vectors for which no meta-component admits an improved match.  Since the overall match can be improved if and only if the match in some meta-component can be improved, the proof is complete.

\end{proof}

\subsection{Proof of Proposition \ref{prop:typeI}} \label{app:proof3}

\begin{center}
\noindent\fbox{%
\vspace{2mm}
\parbox{\textwidth}{%
\vspace{2mm}
Suppose that $Y_{ki}(1)= Y_{ki}(0)$ for all $k,i$.  Then: 
\[
P[T(\mathbf{Z},\mathbf{Y}) > t(\alpha) \mid \mathcal{F}, {\mathcal{M}_{opt}, {\mathcal{U}_{opt}}}] \leq \alpha \quad \quad \forall \alpha \in (0,1).
\]
}
}
\end{center}
\begin{proof}
First, note that it suffices to show that $\mathbf{Z}$ and $\mathbf{Z}'$ are identically distributed conditional on $(\mathcal{F},{\mathcal{M}_{opt}, {\mathcal{U}_{opt}}})$. For any {$\mathbf{z} \in \{0,1\}^N$ (where $N$ is defined as in Section \ref{subsec:setup})}, 
the following is true:
\begin{align*}
\mathbb{P}(\mathbf{Z} =\mathbf{z}  \mid {\mathcal{M}_{opt}, {\mathcal{U}_{opt}}}, \mathcal{F}) 
&= \frac{\mathbb{P}({\mathcal{M}_{opt}, {\mathcal{U}_{opt}}} \mid \mathbf{Z} = \mathbf{z}, 
\mathcal{F})\mathbb{P}(\mathbf{Z}=\mathbf{z} \mid 
\mathcal{F})}{\mathbb{P}({\mathcal{M}_{opt}, {\mathcal{U}_{opt}}}\mid 
\mathcal{F}) }\\
&= \frac{\mathbb{P}({\mathcal{M}_{opt}, {\mathcal{U}_{opt}}}\mid \mathbf{Z} = \mathbf{z}, 
\mathcal{F})\mathbb{P}(\mathbf{Z}=\mathbf{z} \mid 
\mathcal{F})}{
\sum_{\tilde{\mathbf{z}}\in {\{0,1\}^N}}\mathbb{P}({\mathcal{M}_{opt}, {\mathcal{U}_{opt}}} \mid \mathbf{Z} = \tilde{\mathbf{z}}, 
\mathcal{F})\mathbb{P}(\mathbf{Z} = \tilde{\mathbf{z}}\mid 
\mathcal{F})}.
\end{align*}
Since $\mathcal{M}_{opt}$ {and $\mathcal{U}_{opt}$ are deterministic functions} of $\mathbf{Z}$ 
and $\mathcal{F}$, the conditional probability statements {about $\mathcal{M}_{opt}, {\mathcal{U}_{opt}}$ }
 are all equal to one or zero.  Therefore, we can simplify as follows:
\begin{align*}
\mathbb{P}(\mathbf{Z} =\mathbf{z}  \mid {\mathcal{M}_{opt}, {\mathcal{U}_{opt}}}, \mathcal{F}) &= \left\{ \begin{array}{cc} 
\frac{\mathbb{P}(\mathbf{Z}=\mathbf{z} \mid \mathcal{F})}{
\sum_{\tilde{\mathbf{z}}: \mathbb{P}({\mathcal{M}_{opt}, {\mathcal{U}_{opt}}} \mid \mathbf{Z} = {\widetilde{\mathbf{z}}}, 
\mathcal{F}) = 1}\mathbb{P}(\mathbf{Z} = \tilde{\mathbf{z}}\mid 
\mathcal{F})} & \mathbb{P}({\mathcal{M}_{opt}, {\mathcal{U}_{opt}}} \mid \mathbf{Z} = \mathbf{w}, 
\mathcal{F}) = 1 \\
0 & \mathbb{P}({\mathcal{M}_{opt}, {\mathcal{U}_{opt}}}\mid \mathbf{Z} = \mathbf{w}, 
\mathcal{F}) = 0. \end{array} \right.
\end{align*}
By Proposition \ref{prop:unmatched}, Algorithm 2 identifies exactly the set of permuted treatment vectors $\tilde{\mathbf{z}}$ 
such that the right-hand side is positive.  Furthermore, by results in \citet{pimentel2023covariate} and construction of the permutation probabilities using true propensity scores, {among these vectors $\widetilde{\mathbf{z}}$ we have} $\mathbb{P}(\mathbf{Z} = \tilde{\mathbf{z}} \mid 
\mathcal{F}) = \mathbb{P}(\mathbf{Z}' = \tilde{\mathbf{z}} \mid 
\mathcal{F})$, 
 so we may replace the former type of probability statements  in the right-hand side with the latter.  This suffices to show that that $\mathbf{Z}$ and $\mathbf{Z}'$ are identically distributed conditional on $({\mathcal{M}_{opt}, {\mathcal{U}_{opt}}}, \mathcal{F})$, and the proof is complete.
\end{proof}

\end{document}